\documentclass[journal]{IEEEtran}

\usepackage{graphicx}
\usepackage{amssymb,amsmath}
\usepackage{amsthm}
\usepackage{textcomp}       
\usepackage{indentfirst} 
\usepackage{url}
\usepackage{color}
\usepackage{marvosym}
\usepackage{algpseudocode}
\usepackage[font=small]{caption}
\usepackage{tabularx}
\usepackage{multirow}
\usepackage{hhline}
\usepackage[table]{xcolor}
\usepackage{cite}
\usepackage{setspace}
\usepackage{afterpage}
\usepackage{enumitem}
\usepackage{braket}
\usepackage{bbm}
\usepackage{bm}
\usepackage{subfigure}
\usepackage{booktabs}
\usepackage{lipsum}

\usepackage{comment}
\usepackage{siunitx}
\usepackage{tcolorbox}
\tcbset{colback=black!10!white,colframe=black!50!white,boxsep=1pt, boxrule=0.3pt}

\usepackage[]{todonotes}
\presetkeys%
    {todonotes}%
    {inline}{}

\newcommand{\success}{\text{succ}}
\newcommand{\collision}{\text{coll}}
\newcommand{\idle}{\text{idle}}

\newcommand{\tx}{\text{tx}}
\newtheorem{theorem}{Theorem}
\newtheorem{corollary}{Corollary}
\newtheorem{definition}{Definition}

\newtheorem{lemma}{Lemma}

\DeclareMathOperator*{\argmax}{arg\,max}

\title{ORLA/OLAA: Orthogonal Coexistence of LAA and WiFi in Unlicensed Spectrum}
\author{Andres Garcia-Saavedra, Paul Patras, Victor Valls, Xavier Costa-Perez, and Douglas J. Leith\\
\thanks{A. Garcia-Saavedra and X. Costa-Perez are with NEC Laboratories Europe; P. Patras is with University of Edinburgh; V. Valls and D. J. Leith are with Trinity College Dublin.}
}

\begin{document}
\maketitle

\begin{abstract}

Future mobile networks will exploit unlicensed spectrum to boost capacity and meet growing user demands cost-effectively.
The 3GPP has recently defined a Licensed-Assisted Access (LAA) scheme to enable global Unlicensed LTE \mbox{(U-LTE)} deployment, aiming at ($i$) ensuring fair coexistence with incumbent WiFi networks, i.e., impacting on their performance no more than another WiFi device, and ($ii$) achieving superior airtime efficiency as compared to WiFi. In this paper we show the standardized LAA fails to simultaneously fulfill these objectives, and design an alternative orthogonal (collision-free) listen-before-talk coexistence paradigm that provides a substantial improvement in performance, yet imposes no penalty on existing WiFi networks. We derive two LAA optimal transmission policies, ORLA and OLAA, that maximize LAA throughput in both asynchronous and synchronous (i.e., with alignment to licensed anchor frame boundaries) modes of operation, respectively. 
We present a comprehensive performance evaluation through which we demonstrate that, when aggregating packets, IEEE 802.11ac WiFi can be more efficient than 3GPP LAA, whereas our proposals can attain 100\% higher throughput, without harming WiFi. We further show that long U-LTE frames incur up to 92\% throughput losses on WiFi when using 3GPP LAA, whilst ORLA/OLAA sustain $>$200\% gains at no cost, even in the presence of non-saturated WiFi and/or in multi-rate scenarios.
\end{abstract}


\IEEEpeerreviewmaketitle

\section{Introduction}\label{sec:intro}

\IEEEPARstart{5G} radio access network (RAN) architects actively seek to augment mobile systems with inexpensive spectrum in order to boost network capacity and meet growing user demand in a cost-effective manner.
License-exempt \SI{5}{GHz} U-NII channels, currently exploited almost exclusively by WiFi deployments, are of particular interest to the 3GPP 
community~\cite{tr36.889, 3gpp-rel13}, which is pursuing LTE standardization in the unlicensed arena (\mbox{U-LTE}). 
However, the substantial differences between incumbent WiFi, which employs a listen-before-talk (LBT) contention-based multiplexing protocol (CSMA/CA), and LTE, which is inherently a scheduled paradigm, makes the design of U-LTE channel access protocols particularly challenging~\cite{valls-commlet, canocoexistenceton}. The fundamental question facing U-LTE design is \emph{how to exploit this uncharted spectrum efficiently whilst playing fair with native technologies?}

Initial U-LTE solutions employ a Carrier Sensing and Adaptive Transmission (CSAT) scheme based on channel selection and time-based duty cycling~\cite{csat}. The simplicity of this approach has enabled a short time to market in some countries (USA, Korea, India), but it is incompatible with LBT regulations in regions such as Europe and Japan~\cite{etsi-rule}. More recent 3GPP specifications put forward an LBT-based solution named Licensed-Assisted Access (LAA), to address this deficiency~\cite{3gpp-rel13}. The LBT flavor of LAA is similar to WiFi CSMA/CA, thereby enabling global U-LTE deployment. The impact of different U-LTE approaches on the performance of WiFi technology native to unlicensed bands is continuously studied by the research community~\cite{hajmohammad2013unlicensed, liu2014small, guan4cu, gomez2016srslte, capretti2016}.\footnote{We refer the interested reader to \S\ref{sec:related} for a review of related research.}

\begin{figure}[t!]
      \centering
      \includegraphics[trim = 0mm 0mm 0mm 0mm, clip, width=0.9\linewidth]{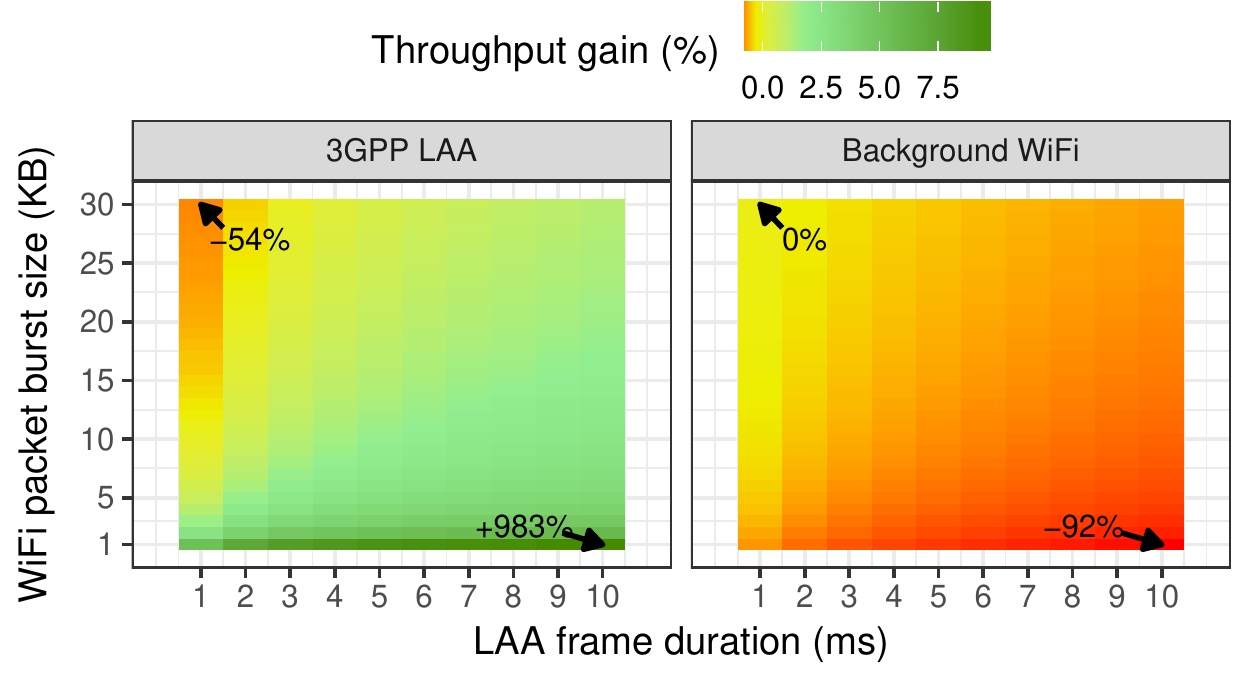}
       \vspace{-3mm}	
      \caption{Performance gain/loss experienced by one coexisting asynchronous 3GPP LAA transmitter and 5 WiFi transmitters, over a WiFi network with 6 transmitters. Both technologies transmit as much data as possible with the same modulation (64-QAM) and contention window configuration. Simulation results.}
       \vspace{-2mm}
      \label{fig:motivation}
\end{figure}

In this paper we argue that the 3GPP LAA scheme is unable to simultaneously meet the following key design criteria: 
($i$) avoid harming the performance of preexisting WiFi wireless nodes (WNs), and ($ii$) provide superior medium access control (MAC) protocol efficiency as compared to WiFi. We illustrate this with an example in Fig.~\ref{fig:motivation}, where 1 LAA WN coexists with 5 backlogged IEEE 802.11ac~\cite{IEEE80211ac} WNs. We plot the MAC throughput gain/loss attained by LAA and WiFi relative to the performance that would have been achieved if an additional WN was added to the network (i.e. 6-WiFi WNs) instead of the LAA WN.  Results are shown for a range of LAA frame durations (x-axis) and WiFi packet burst sizes (y-axis). 
Observe that LAA achieves its highest throughput gain when the background WiFi WNs are most penalized (bottom right corners). {This is due to the fact that LAA transmits data for \SI{10}{\milli \second} upon winning a successful transmission opportunity whereas a WiFi WN only transmits for \SI{62}{\micro \second} ($1$~KB at 130~Mb/s).  This leads to WiFi having a high relative overhead (inter-frame spaces and idle backoff periods)}. Further, observe that in the region where LAA is relatively harmless to WiFi (top left corners), it experiences barely any throughput performance benefit (if not a loss). 

In essence, since the LAA access procedure fundamentally resembles that of WiFi's CSMA/CA it does have the ability to provide fairness; however, increases in MAC efficiency are achieved at the cost of greatly penalising WiFi devices.
Clearly, despite these preliminary results (a thorough evaluation campaign is presented in \S\ref{sec:results}), U-LTE can provide enhanced PHY-layer efficiency (e.g., more robust error recovery mechanisms) and operational advantages (e.g., common radio resource management with licensed LTE), which may sustain the appeal of LAA irrespective of its (inefficient) MAC-layer coexistence mechanism. In this paper we argue, however,  that U-LTE's MAC-layer efficiency can (and should) be substantially improved \emph{without compromising incumbent WiFi}.
We therefore revisit U-LTE's MAC coexistence design, making use of the observation that if the airtime used by WiFi and U-LTE were decoupled by  \emph{eliminating inter-technology collisions}, we could build U-LTE access schemes employing transmission strategies that are demonstrably harmless to WiFi and achieve higher data rates compared to CSMA/CA-based alternatives for cellular access to unlicensed spectrum. These include LTE-WLAN aggregation (LWA or LTE-H) \cite{lwa} and 3GPP LAA. {In summary, we make the following contributions}:
\begin{itemize}
\item We design an orthogonal (collision-free) LBT-based access protocol that coexists with WiFi truly seamlessly;
\item We construct optimal LBT transmission policies that suit a variety of network conditions, including heterogeneous traffic loads and link rates;
\item We devise optimal policies for synchronous LBT technologies, constrained to transmit during fixed intervals;
\item We undertake a thorough simulation campaign to evaluate the proposed policies and demonstrate gains in LBT throughput of more than 200\% with no negative impact on WiFi;
\item We give practical guidelines to aid the implementation of the mechanisms we propose on off-the-shelf hardware.
\end{itemize}

The rest of the paper is organized as follows. \S\ref{sec:preliminaries} reviews fundamental technical concepts underpinning the design of the proposed U-LTE coexistence scheme and in \S\ref{sec:mechanism} we present their operation principles. In \S\ref{sec:orla} and \S\ref{sec:olaa} we introduce optimal LBT transmission policies that simultaneously accomplish fair resource sharing with WiFi WNs and throughput maximization. We evaluate via simulations our access procedures and transmission policies in \S\ref{sec:results}, comparing against 3GPP's LAA scheme. We discuss practical details that can facilitate the implementation of our solutions in \S\ref{sec:practical}. Finally, \S\ref{sec:related} reviews related work and \S\ref{sec:conclusions} draws concluding remarks.

\section{Preliminaries}\label{sec:preliminaries}
We begin by reviewing two technical aspects fundamental to the design of the coexistence mechanism we propose, namely ($i$)~LBT regulatory constraints in \SI{5}{GHz} bands, and ($ii$)~the Distributed Coordinated Function (DCF) that governs communication in IEEE 802.11 WiFi. 
We focus on the \text{ETSI 301 893} regulation \cite{etsi-rule}, as it is the most restrictive and a solution compliant with this will be widely deployable. 
To stress that the coexistence mechanism we propose extends to any technology that seeks operation on \SI{5}{GHz} channels, hereafter we refer to an LTE-LAA cell as an LBT WN.

\subsection{Listen Before Talk (LBT)} 
\label{sec:lbt-fbe}
The \text{ETSI 301 893} standard \cite{etsi-rule} specifies that load based equipment (LBE) shall implement LBT following the Clear Channel Assessment (CCA) mode using energy detection.\footnote{{Load base equipment (LBE) does not follow a fixed TX/RX pattern, but is driven by demand~\cite{etsi-rule}. In contrast, frame based equipment (FBE) transmits at fixed intervals. 
Since our goal is to opportunistically exploit the medium for LTE transmissions, we work with the LBE paradigm.}}
Energy detection refers to observing the operating channel for a pre-defined duration and determining whether the energy level sensed exceeds a sensitivity threshold. In such circumstances, the channel is regarded busy and transmission deferred. ETSI mandates that CCA assertion can be performed in accordance with the IEEE 802.11 standard's provisions~\cite{IEEE80211ac}. Alternatively, minimum requirements should be met, as defined by two channel access options stipulated in \cite{etsi-rule}.

In our design, we consider LBT WNs that comply with clause~18 of IEEE 802.11~\cite{IEEE80211ac}, which requires OFDM transmitters to identify a busy channel within \SI{4}{\micro \second}. This also requires a WiFi station operating in the 5~GHz band to observe the channel idle for at least DIFS = \SI{34}{\micro \second}, before attempting to transmit. We further note that 802.11 data frames and acknowledgements (ACKs) are separated by SIFS = \SI{16}{\micro \second}. 
We propose allowing the LBT WN to attempt transmission following an 802.11 frame exchange, immediately after the channel is sensed idle for an LTE interframe space LIFS = \SI{20}{\micro \second} that we introduce.  
This ensures an 802.11 frame exchange will not be interrupted, while avoiding potential collisions with co-existing 802.11 stations (which may transmit immediately after DIFS, if initializing random back-off counters with zero).
We summarize the 802.11 channel access procedure next.

\subsection{IEEE 802.11 MAC protocol (WiFi)}

WiFi medium access is regulated by the IEEE 802.11 DCF, which performs CSMA/CA with Binary Exponential Backoff (BEB). While a detailed description is given in \cite{bianchi2000performance}, we include here a brief description here for completeness. An IEEE 802.11 network divides time into MAC slots and a station transmits after observing $S_m$ idle slots, where $S_m$ is a random variable selected uniformly at random from $\{0,1,\dots, 2^m CW_{\text{min}} -1 \}$ where $CW_{\text{min}} \in \mathbb N$ is the minimum contention window and $m=0,1,2,\dots$ is the number of successive collisions experienced by the station. After a successful transmission $m$ is set to $0$. IEEE 802.11 defines a parameter $CW_{\text{max}}$ that limits the expected number of idle slots a station has to wait after $m$ successive collisions, i.e., $2^{\bar{m}} CW_{\text{min}} = CW_{\text{max}}$ for $m \ge \bar m$.

Two key features of WiFi systems relevant to our work are: ($i$) each 802.11 packet includes in the header information regarding the duration of the transmission, i.e., upon correct reception of a packet, a station knows the duration for which the channel will be busy; ($ii$) after a successful transmission all stations in the network wait for an Arbitration Inter-Frame Spacing (AIFS) time of at least \SI{34}{\micro \second}.\footnote{The AIFS time corresponds to the DCF Inter-Frame Space (DIFS) in DCF-based devices, and in the  \SI{5}{GHz} bands it has a duration of at least \SI{34}{\micro \second}.} That is, importantly, after each successful transmission there will be at least \SI{34}{\micro \second} during which the channel is free of WiFi transmissions.

\section{Orthogonal Airtime Coexistence}\label{sec:mechanism}

Our first objective is to design a MAC-layer protocol for a non-WiFi LBT-based WN (e.g. LTE-LAA cells) that fulfills the following criteria: ($i$) pose no harm to preexisting WiFi wireless networks (WNs), and ($ii$) improve MAC efficiency over WiFi WNs. 
To formalize these objectives, we first introducing the following definitions.

\begin{definition}\label{def:wifi}
We let ``WiFi WN'' define a wireless transmitter operating in unlicensed bands following the IEEE 802.11 specification. 
\end{definition}
\noindent By the above definition, WiFi WNs observe a slotted channel, where slot duration varies as the channel can be idle, occupied by a transmission (whose length depends on the payload size and bit rate), or contain a collision between two or more simultaneous WiFi WN transmissions. 

\begin{definition}\label{def:lbt}
We let ``LBT WN'' define a non-WiFi wireless transmitter operating in unlicensed bands using a Listen-before-Talk (LBT) access mechanism. ``Synchronous LBT WN'' refers to LBT WNs constrained to transmit in predefined fixed timeslots, and ``asynchronous LBT WN'' those which are not.
\end{definition}
\noindent We note that LTE-LAA may require synchronization with TDMA LTE framing in anchor licensed bands (synchronous LTE-LAA) or not, as in the case of MulteFire~\cite{multefire} (or asynchronous LTE-LAA). Importantly, we let ``WiFi/LBT WN'' refer to \emph{any} transmitter using unlicensed wireless channels \emph{including the WiFi AP in downlink scenarios and any WiFi station in uplink communication} so in general we will avoid making explicit reference to uplink or downlink.

Our coexistence approach builds upon the key observation that the minimum duration of an DIFS/AIFS (\SI{34}{\micro \second}) is longer than the CCA minimum time (\SI{20}{\micro \second}) specified by ETSI.
More specifically, to avoid inter-technology collisions, we assume that an LBT WN can acquire the channel if the medium is sensed idle for a LIFS = \SI{20}{\micro \second} (LBT inter-frame space) duration, a timing constant we introduce. Note that SIFS $<$ LIFS $<$ PIFS ($<$ AIFS/DIFS), which means LBT transmissions take priority, but cannot interrupt ongoing data--ACK exchanges (which are separated by a SIFS) in concurrent WiFi transmissions. Thus, \textbf{allowing idle channel acquisition after LIFS enables an LBT system to opportunistically exploit orthogonal collision-free airtime in unlicensed spectrum.} 

\begin{figure}
 \includegraphics[width=\columnwidth]{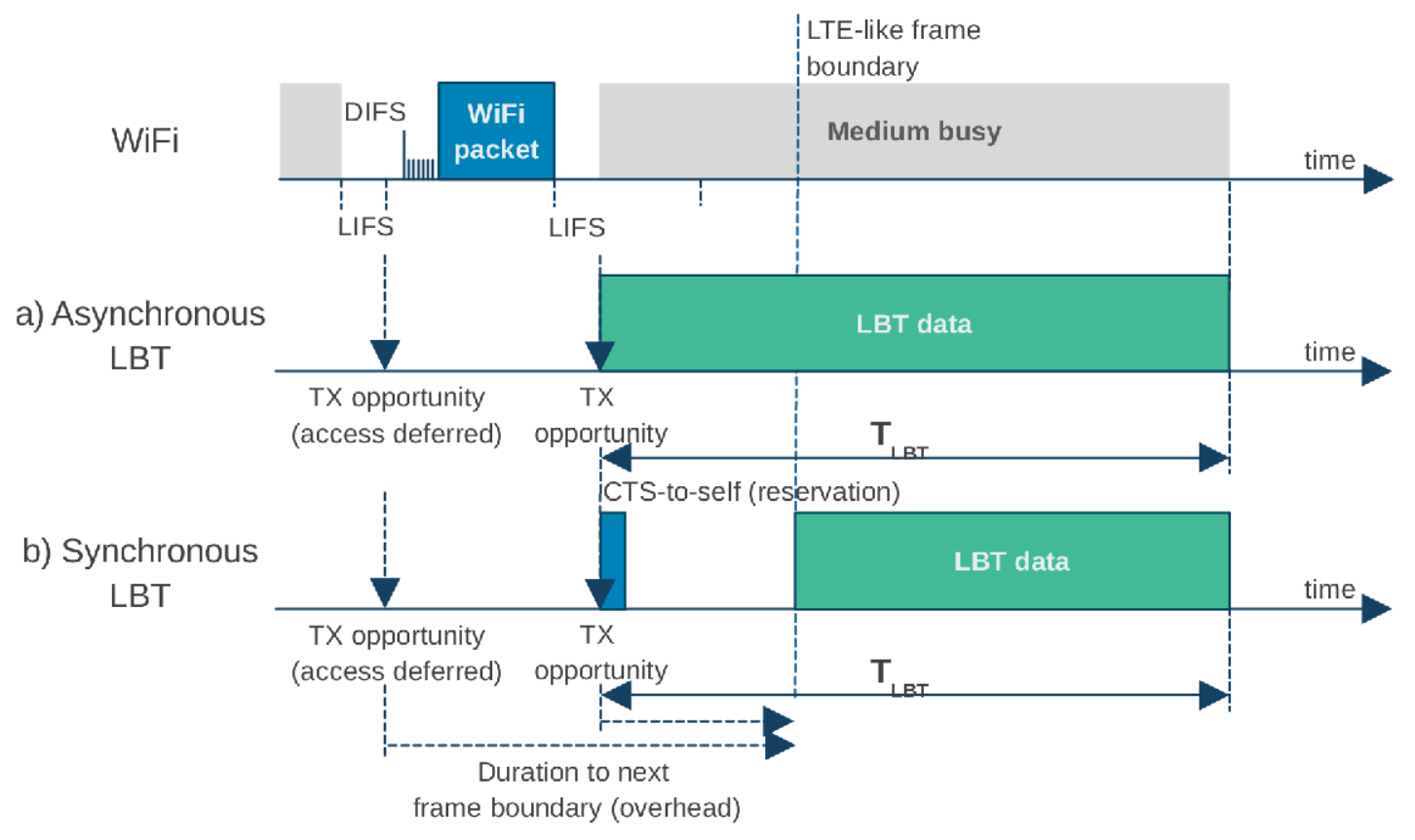}
 \caption{\small Proposed protocol operation with a WiFi WN and an (a)synchronous LBT WN. After a busy period with channel occupied by WiFi, LBT WNs have a transmission opportunity. In this example, the first transmission opportunity is skipped due to e.g. coexistence policy; the second opportunity is taken. In case of asynchronous access, the LBT WN delivers data immediately. In case of synchronous access, the channel is reserved (via a CTS-to-self message) and data is sent only at the boundary of the next frame.}
 \label{fig:protocol}
 \vspace{-2mm}
\end{figure}

To support \emph{synchronous LBT WNs} that need to postpone data transmissions until the beginning of a frame, we employ the 802.11 compliant CTS-to-self mechanism. This enables the synchronous LBT WN to reserve the channel and so ensure this frame alignment. Specifically, if the synchronous LBT WN senses the channel idle for LIFS, it can either send a CTS-to-self to reserve the medium until the next frame, or it can defer access until the next opportunity. In either case, when an asynchronous/synchronous LBT WN decides to acquire the channel, it will hold the channel for a fixed duration $T_{LBT}$, as permitted by regulations for both Frame Based Equipment (FBE) or Load Based Equipment (LBE) \cite{etsi-rule}. We exemplify the proposed \emph{modus operandi} in Fig.~\ref{fig:protocol}. 

Note that with the above scheme, transmissions by LBT and WiFi WNs will never collide, hence an LBT WN does not affect the transmission attempt probability of the stations in a WLAN.\footnote{This is true under certain assumptions which we discuss next and in \S\ref{sec:practical}.}  Therefore, the shared channel can be regarded as divided into two \emph{orthogonal} airtimes. In addition, since an LBT WN will always sense the channel \emph{idle} immediately after a WiFi transmission, the policies specified for FBE and LBE as to how to perform another CCA when the channel is sensed \emph{busy} are irrelevant to this work. Note however that collisions amongst LBT WNs could occur. In the rest of the paper we will assume there is only one LBT WN, which is in line with 3GPP~\cite{3gpp-rel13} and related work (e.g. \cite{7247522,canocoexistence,canocoexistence2,6655388}), and corresponds to the case of downlink offloading with cellular operators using different channels to circumvent interference. 
We discuss the case with multiple LBT WN sharing the same channel in \S\ref{sec:practical}.

Making use of the LBT access protocol introduced above, our second objective it to find \emph{optimal transmission policies for both synchronous and asynchronous LBT.} To this end, we finalize this section by formalizing our design criteria with two additional definitions.

\begin{definition}[\bf LBT transmission policy]
Index the LIFS transmission opportunities (the times when a station using our LBT access protocol could potentially transmit) by $\omega=1, 2, \ldots$ and gather these indices into set $\Omega \subset \mathbb{N}$. Then $\Pi\in 2^\Omega$ defines an \emph{LBT transmission policy}, where $2^\Omega$ denotes the set
of all possible subsets (the superset) of $\Omega$. An LBT station using policy $\Pi$ transmits one burst of duration $T_{LBT}$ at each LIFS transmission opportunity $\omega \in \Pi$.

\end{definition}
\begin{definition}[\bf Optimal policies]\label{def:fairness}
 An airtime-optimal  policy $\Pi^\dagger$ and a throughput-optimal policy $\Pi^\ddagger$ maximize, respectively, the channel time held by an LBT WN and the throughput of an LBT WN, while preserving (at least) the same amount of channel time each coexisting Wi	Fi WN would have if the LBT WN implemented a WiFi access protocol.   
\end{definition}
In what follows, we construct optimal LBT transmission policies for a range of scenarios.

\section{ORLA: Orthogonal Random LBT Unlicensed Access Strategy} \label{sec:orla}

In this section we devise optimal LBT transmission policies for asynchronous LBT WNs coexisting with WiFi. 
We first address airtime maximization under circumstances where WiFi contenders are backlogged (saturated) and operate with the same PHY bit rate (homogeneous), then focus on maximizing airtime usage when WiFi is lightly loaded (non-saturated) and WNs employ dissimilar rates (heterogeneous links).

\subsection{Homogeneous and Saturated Conditions}\label{sec:setup}

The scope of our work is MAC layer enhancement, therefore we assume ideal physical layer conditions, i.e LBT and WiFi WNs are within carrier sensing range of each other (no hidden terminals), no capture effect, and perfect PHY rate control mitigates channel errors. Thus losses due to fading are negligible,  whilst we consider practical multi-rate WiFi operation in \S\ref{sec:non-sat}. 
These assumptions not only ensure mathematical tractability but help us keep a focused analysis. We assume LBT WNs equipped with an off-the-shelf WiFi interface for channel sensing and medium access reservation purposes, in addition to the native LTE modem. We note this is common practice in coexistence mechanism design (see e.g.~\cite{7247522}) and we discuss in detail the practical implications of the above assumptions in \S\ref{sec:practical}. 

Consider a scenario with $n$ saturated WiFi WNs, i.e., each WN always has data ready for transmission. It is well known that under these conditions the transmission attempt probability $\tau^{(n)}_i$ of WN $i$ in random MAC slot can be related to the conditional collision probability $p^{(n)}$. Given the homogeneous load assumption, $\tau^{(n)} = \tau^{(n)}_i$ for all $i \in \{1,\dots,n\}$ (we relax this assumption later), $\tau^{(n)}$ can be computed by solving the following system of non-linear equations~\cite{1208922}:
\[
\begin{cases}
 \tau^{(n)} = \frac{2(1-2p^{(n)})}{(1-2p^{(n)})(CW_{\text{min}} + 1) + p^{(n)}CW_{\text{min}}(1- (2p^{(n)})^{m})},\\
 p^{(n)} = 1- (1-\tau^{(n)})^{n-1},
\end{cases}
\]
where recall $CW_{\min}$ is the contention window minimum parameter and $m$ is the backoff stage.
Then, the probability that a MAC slot is idle is given by the probability that none of the WNs transmits, i.e. 
$P^{(n)}_\idle = (1-\tau^{(n)})^n;$
the probability that a slot is occupied by a successful transmission is 
$P^{(n)}_\success =n p^{(n)}_\success$, where $p^{(n)}_\success = \tau^{(n)} (1-\tau^{(n)})^{n-1}$ is the probability that a single WN transmits in a MAC slot. Finally, the probability that a slot is occupied by a collision is given by $P^{(n)}_\collision = 1- P^{(n)}_\idle -  P^{(n)}_\success$ and the probability of a slot being busy is $P^{(n)}_\tx=P^{(n)}_\collision+P^{(n)}_\success$.  
The throughput of a WiFi WN is given by

\begin{align}
s^{(n)} = \frac{p^{(n)}_\success B}{P^{(n)}_\idle \sigma + (1-P^{(n)}_\idle)T} ,
\end{align}where $\sigma$, $B$, and $T$ are the duration of an (idle) MAC slot, the expected number of bits in a transmission, and the duration of a transmission (successful or collision) respectively, which is equal to
\small
\begin{align}
 T = T_\text{PLCP}\!+\!\frac{\! f_{\text{agg}} \left(L_{\text{del}} \!+\! L_{\text{mac-oh}} \!+\! L_{\text{pad}} \right) \!+B}{C}\!+\!\text{SIFS}\!+\!T_\text{ACK}\!+\!\text{DIFS}.\nonumber
\end{align}
\normalsize
SIFS, DIFS, $T_\text{PLCP}$, $L_\text{del}$, and $L_{\text{pad}}$ are PHY layer constants (inter-frame spacing, delimiters, padding), $L_{\text{mac-oh}}$ is the MAC layer overhead (header and FCS), $f_{\text{agg}}$ is the number of packets aggregated in a transmission, $B$ the expected number of data bits transmitted in the burst (payload), and $C$ the PHY bit rate. The duration of an acknowledgement is
\begin{align}
 T_\text{ACK} = T_\text{PLCP} +  \frac{L_{\text{ACK}}}{C_{\text{ctrl}}},\nonumber
\end{align}
where $C_{\text{ctrl}}$ is the bitrate used for control messages.

We aim to obtain the maximum fraction of orthogonal airtime that an LBT WN can use such that the average throughput experienced by a WiFi WN is not degraded more than if another WiFi WN were added to the network. Since LBT transmissions following the access procedure proposed in \S\ref{sec:mechanism} are orthogonal to WiFi transmissions, an LBT WN can be regarded (in terms of airtime) as a WiFi WN that transmits in MAC slots that otherwise would be idle. Then, the LBT airtime can be expressed as
\begin{align}
A _{\text{LBT}}= \rho P^{(n)}_\idle (T'-\sigma), 
\label{eq:airtime}
\end{align}
where $\rho \in [0,1]$ is the fraction of idle slots that would change to busy slots, and $(T' - \sigma):=T_\text{LBT} > 0$ is the duration of an LBT WN transmission, which depends on the LBT mode used (FBE or LBE). Note that the quantity $\rho P^{(n)}_\idle$ is the fraction of {orthogonal} LBT transmissions. With (\ref{eq:airtime}) we can write the throughput experienced by a WiFi WN when an LBT WN uses $A_{\text{LBT}}$ airtime, as follows:
\begin{align}
s^{(n+\text{LBT})} &:=  \frac{p^{(n)}_\success B}{P^{(n)}_\idle \sigma  + P^{(n)}_\tx T +  \rho P^{(n)}_\idle (T'-\sigma)}. \label{eq:txlbtineq}
\end{align}Next, since the throughput of a WiFi WN is non-increasing with the number of WNs, i.e., $s^{(n)} \ge s^{(n+1)}$ for every $n=1,2,\dots$, we have that 
\begin{align}
 s^{(n+1)} = \frac{p^{(n+1)}_\success B}{P^{(n+1)}_\idle \sigma + P^{(n+1)}_\tx T}  \le s^{(n+\text{LBT})}\label{eq:lbtthrouhput} 
\end{align}will always hold, provided $\rho$ in (\ref{eq:txlbtineq}) is sufficiently small. We are interested in finding the value of $\rho$ that makes (\ref{eq:lbtthrouhput}) tight, i.e., maximises the LBT airtime. 
We give the following lemma.

\begin{lemma} \label{th:barrho} 
Consider $n$ homogeneous saturated WNs. Suppose $T,T'  > \sigma$. Then, (\ref{eq:lbtthrouhput}) holds for every $\rho \in [0, \bar \rho]$ with 

\vspace{-0.75em}
{\small
\begin{align}
\bar \rho := 
\left(\frac{T-\sigma}{T'-\sigma}\right) \min \left \{1, \frac{P^{(n+1)}_\tx }{p^{(n+1)}_\success } \frac{p^{(n)}_\success}{P^{(n)}_\idle } - \frac{P^{(n)}_\tx}{P^{(n)}_\idle} \right \}  \label{eq:barrho}
\end{align}}
\end{lemma}
\begin{IEEEproof}
Rearranging terms in (\ref{eq:lbtthrouhput}) with $P_\tx = (1- P_\idle)$ and $A = \rho P^{(n)}_\idle (T'-\sigma)$ we have that
{\begin{align}
\frac{p^{(n)}_\success }{p^{(n+1)}_\success }
& \ge
\frac{P^{(n)}_\idle ( \sigma - T) +   T  + \rho P^{(n)}_\idle (T'-\sigma) } {P^{(n+1)}_\idle (\sigma - T) +   T } .
\end{align}}Further rearranging we obtain that
{\begin{align*}
&  \rho P^{(n)}_\idle (T'-\sigma)  \\
&  \le \frac{p^{(n)}_\success }{p^{(n+1)}_\success }( P^{(n+1)}_\idle ( \sigma - T) +   T)   - P^{(n)}_\idle ( \sigma - T) -    T,  \\
&  =  T \left(  \frac{p^{(n)}_\success }{p^{(n+1)}_\success }-1 \right) + \left( P^{(n)}_\idle - \frac{p^{(n)}_\success }{p^{(n+1)}_\success } P^{(n+1)}_\idle \right)(T - \sigma) ,  
\end{align*}}and dividing by $P^{(n)}_\idle (T' - \sigma)$ yields
{\begin{align*}
\rho & \le  \frac{T}{P^{(n)}_\idle (T'-\sigma)} \left(  \frac{p^{(n)}_\success }{p^{(n+1)}_\success }-1 \right) 
+  \left(1 - \frac{p^{(n)}_\success }{p^{(n+1)}_\success } \frac{P^{(n+1)}_\idle}{P^{(n)}_\idle}\right).
\end{align*}}Now fix $T'= T$ and see that since $T / (T - \sigma) > 1$ we have
{\begin{align}
\rho & \le  \frac{1}{P^{(n)}_\idle } \left(  \frac{p^{(n)}_\success }{p^{(n+1)}_\success } - 1   + {P^{(n)}_\idle}  - \frac{p^{(n)}_\success }{p^{(n+1)}_\success } {P^{(n+1)}_\idle}\right), \notag \\
& \le \frac{P^{(n+1)}_\tx }{p^{(n+1)}_\success } \frac{p^{(n)}_\success}{P^{(n)}_\idle } - \frac{P^{(n)}_\tx}{P^{(n)}_\idle} ,  \label{eq:prerho}
\end{align}}where in (\ref{eq:prerho}) we have used the fact that $1-P_\idle  = P_\tx$ and $\rho \le 1$.  Finally, when $T'\ne T$, since all that matters is the total airtime $A_{\text{LBT}}$ given in (\ref{eq:airtime}), if we multiply (\ref{eq:prerho}) by $(\frac{T-\sigma}{T'-\sigma})$ the stated result follows.
\end{IEEEproof}
With Lemma \ref{th:barrho} we can obtain the fraction of orthogonal/successful LBT transmissions ($\rho P^{(n)}_\idle$) of expected duration $T_{\text{LBT}} =T'-\sigma$ that can be accommodated in order to be compliant with our coexistence criterion.
Importantly, the bound in (\ref{eq:barrho}) depends on $P_\tx^{(n+1)}$ and $p_\success^{(n+1)}$, however, in saturation conditions a very good approximation of these values can be easily obtained \cite{bianchi2000performance}.

Observe that we can write 
\begin{align*}
& P^{(n)}_\idle \sigma + (1-P^{(n)}_\idle)  T + \rho P^{(n)}_\idle (T' - \sigma) \\
& =P^{(n)}_\idle \sigma + (1-P^{(n)}_\idle) (T  +  \rho P^{(n)}_\idle \frac{T' - \sigma}{1-P^{(n)}_\idle}) \\
& = P^{(n)}_\idle \sigma + (1-P^{(n)}_\idle) (T  + \pi(T'-\sigma)),
\end{align*}
where 
\begin{align}
 \pi(\rho) := \min \left\{ 1, \rho \frac{ P^{(n)}_\idle}{1- P^{(n)}_\idle } \right \}.
 \label{eq:optpi}
\end{align}
That is, an LBT WN will be compliant with our coexistence criterion as long as it takes a fraction $\pi(\bar \rho)$ of collision-free LIFS opportunities after a busy slot (successful or collision), where $\bar \rho$ is given in Lemma~\ref{th:barrho}. In this way, we can establish the following optimal transmission policy, hereafter referred to as ``Orthogonal Random LBT Unlicensed Access'' (ORLA):

\begin{theorem}[\bf ORLA transmission policy]\label{th:orla}
 Consider a policy $\Pi^{\text{ORLA}}$ by which an (LBT) WN initiates transmissions for a fixed duration $T_{LBT}$ after a LIFS opportunity $\omega\in\Omega$ with probability $\pi( \rho)$:
\begin{align}
\Pi^{ORLA}(\rho) := \{ \omega \in \Omega \mid \mathcal{U}_\omega (\rho) = 1 \},\nonumber
\end{align}
where $\mathcal{U}_\omega (\rho), \omega \in \Omega$ are random variables taking values 0 or 1 such that $\mathrm{Pr}(\mathcal{U}_\omega (\rho) = 1) = \pi(\rho)$.
In an homogeneous scenario where all WiFi WNs are saturated and have the same channel access configuration, $\Pi^{ORLA}(\bar\rho)$ is an airtime-optimal policy $\Pi^\dagger$ and a throughput-optimal transmission policy $\Pi^\ddagger$ for synchronous LBT and asynchronous LBT, respectively. 
\end{theorem}
\begin{proof}
 The proof follows directly from Lemma~\ref{th:barrho}.
\end{proof}

\subsection{Non-saturation and Heterogeneous Conditions}
\label{sec:non-sat}
Next we generalize the results above to heterogeneous conditions in terms of WiFi packet arrival rates and link qualities. Recall that since cellular deployments work permanently on licensed frequencies, we only investigate the performance of supplemental downlink services exploiting the unlicensed band for best effort data transfers. As such, we still consider backlogged LBT WN newcomers (i.e. always having data to transmit) and study the performance of the proposed system with ($i$) practical multi-rate WiFi operation, thus focusing~on airtime instead of throughput fairness, while \emph{(ii)} utilizing additional airtime released by WiFi WNs with finite loads.

We proceed by extending the ORLA transmission policy to guarantee that the aggregate channel time of $n$ (non-saturated) WiFi WNs when we add a \emph{saturated} WiFi WN, $A^{(n+1)}$, remains constant or larger than the aggregate channel time of $n$ \emph{saturated} WiFi WNs when we add the LBT WN, $A^{(n_{\text{sat}}+\text{LBT})}$. Formally,
\begin{align}\label{eq:a_n+1}
 A^{(n+1)} = \frac{ \sum\limits_{i=1}^{n} p^{(n+1)}_{\success,i} T_{s,i}}{T_{\text{slot}}^{(n+1)}}
\end{align}
where 
\begin{align}
 T_{\text{slot}}^{(n+1)} =  {P^{(n+1)}_\idle}\sigma + \sum\limits_{i=1}^{n+1} p^{(n+1)}_{\success,i}T_{s,i} + P^{(n+1)}_\collision T_c, \nonumber
\end{align}
with $p^{(n+1)}_{\success,i}$ being the probability that WN $i$ transmits successfully in a MAC slot, $T_{s,i} = T_\text{PLCP}+[f_{\text{agg,i}} (L_{\text{del}} \!+\! L_{\text{mac-oh}} \!+\! L_{\text{pad}}) \!+B_i]/C_i+\text{SIFS}+T_\text{ACK}+\text{DIFS}$ the duration of a slot when WN $i$ transmits successfully, and $T_c$ the time the channel remains busy during a collision. Note that the numerator in~\eqref{eq:a_n+1} sums over the $n$ (non-saturated) WiFi WNs and not over all the WNs in the system.

To compute the WiFi WNs' transmission attempt rates $\tau_i, i = 1,\ldots,n$, we first rewrite the conditional collision probability $p^{(n)}_i$ that the frames transmitted by WN $i$ experience:
\begin{equation}
  p^{(n)}_i = 1 - \prod_{k=1,k \neq i}^{n} \left(1-\tau^{(n)}_k\right).
\label{eq:pcol}
\end{equation}

We use a renewal-reward approach to model the WiFi BEB scheme in the presence of different packet arrival rates. To avoid notation clutter, we drop the $i$ and ${(n)}$ indexes when there is no scope for confusion. The transmission attempt rate of a WiFi contender can be thus expressed as 
\[
\tau=\frac{E[A]}{E[S]},
\]
where $E[A]$ is the expected number of attempts to transmit a packet burst and $E[S]$ is the expected number of slots used during back-off, which we compute as follows:
\[
 E[A]=1+p+p^2+...+p^M, 
\]
\[
 E[S]=t_\text{\text{idle}}+b_0+p b_1+p^2 b_2+...+p^M b_M,
\] 
where $M$ is the maximum number of retries (which we assume equal to the maximum back-off stage $\bar{m}$) and $b_m$ is the mean length of back-off stage $m$ expressed in slots. $t_\text{\text{idle}}$ is the mean idle time that a contender waits for new content after a transmission. Thus, we can express the transmission attempt rate of a WiFi transmitter as
\begin{equation}
\tau=\frac{E[A]}{E[S]}=\frac{1+p+p^2+...+p^M}{t_i+b_0+p b_1+p^2 b_2+...+p^M b_M}.
\label{eq:tau}
\end{equation}
We apply the above to relate $\tau_i^{(n)}$ to $p_i^{(n)}, \forall i$. Neglecting post-backoff and assuming no buffering, we can write 
\begin{equation}
 t_\text{\text{idle}} = q(1+2(1-q)+3(1-q)^2+...) = \frac{1}{q},
\end{equation}
where $q$ is the probability that a new frame arrives in a uniform slot time $T_\text{slot}$. Note that, assuming Poisson arrivals, we can related $q$ to a WN offered load $\lambda$ as $\lambda = -\log(1-q)/T_\text{slot}$.

Analogously, 
\begin{align}
 A^{(n_{\text{sat}}+\text{LBT})} = \frac{ \sum\limits_{i=1}^{n_{\text{sat}}} p^{(n_{\text{sat}})}_{\success,i} T_{s,i}}{ T_{\text{slot}}^{(n_{\text{sat}}+\text{LBT})} },
\end{align}
where 
\begin{align}
 T_{\text{slot}}^{(n_{\text{sat}}+\text{LBT})} = T_{\text{slot}}^{(n_{\text{sat}})} + \rho P^{(n_{\text{sat}})}_\idle \left ( T'-\sigma \right ).\nonumber
\end{align}
The remaining task is to find $\rho$, such that the following condition is satisfied:
\begin{align}\label{eq:airtimemax}
 A^{(n_{\text{sat}}+\text{LBT})} \ge A^{(n+1)}.
\end{align}
This guarantees that ($i$) pre-existing WiFi WNs satisfy their traffic demands as if a saturated WiFi WN would be added to the system, and ($ii$) the LBT WN maximizes the channel time devoted to transmission. 

Proceeding similarly to Lemma~\ref{th:barrho}, we obtain the following:
\begin{lemma}
\label{lem:bound}
 In a scenario with $n$ WiFi WNs operating with different offered loads and transmission bit rates, which shares the channel with an LBT WN, (\ref{eq:airtimemax}) holds for every $\rho \in [0,\bar{\rho}]$, where $\bar\rho$ is computed as:
 \fontsize{9.5}{11}
 \begin{align}\label{eq:th:barrho}
 \bar{\rho} := \frac{1}{\rho P^{(n_{\text{sat}})}_\idle \left ( T'\!\!-\!\sigma \right )}\left[    \frac{\sum\limits_{i=1}^{n_{\text{sat}}} p^{(n_{\text{sat}})}_{\success,i} T_{s,i} }{\sum\limits_{i=1}^{n} p^{(n+1)}_{\success,i} T_{s,i}} T_{\text{slot}}^{(n+1)} \!\!-\!T_{\text{slot}}^{(n_{\text{sat}})} \right].
 \end{align}
 \normalsize
\end{lemma}
\begin{IEEEproof}
 Rearranging the terms in \eqref{eq:airtimemax} yields $$\frac{\sum\limits_{i=1}^{n_{\text{sat}}} p^{(n_{\text{sat}})}_{\success,i} T_{s,i} }{\sum\limits_{i=1}^{n} p^{(n+1)}_{\success,i} T_{s,i}} \le \frac{T_{\text{slot}}^{(n+1)}}{T_{\text{slot}}^{(n_{\text{sat}})} + \rho P^{(n_{\text{sat}})}_\idle \left ( T'\!\!-\!\sigma \right )}.$$ 
 Then, the result in~\eqref{eq:th:barrho} can be verified by proceeding similarly to the proof of Lemma~\ref{th:barrho}.
\end{IEEEproof}

By applying Lemma~\ref{lem:bound} and estimating the packet arrival rates at the WiFi WNs, the LBT WN can maximize channel utilization, leading to the following extension to Theorem~\ref{th:orla}.
\begin{corollary}[\bf ORLA transmission policy]
Given the ORLA transmission policy $\Pi^{\text{ORLA}}(\rho)$ defined in Theorem~\ref{th:orla} and a heterogeneous scenario where WiFi WNs transmit at different rates and have different loads, $\Pi^{\text{ORLA}}(\bar\rho)$ is an airtime-optimal policy $\Pi^\dagger$ and a throughput-optimal transmission policy $\Pi^\ddagger$ for synchronous LBT and asynchronous LBT, respectively, where $\bar\rho$ is derived with Lemma~\ref{lem:bound}.
\end{corollary}
\begin{proof}
 The proof follows directly from Lemma~\ref{lem:bound}.
\end{proof}

Observe that, when the aggregate WiFi load is low, \mbox{$\pi(\rho) \rightarrow 1$}, which may lead to underutilization of available channel time, as the LBT WN has insufficient LIFS opportunities to transmit. To address this issue, we let the LBT WN transmit $N$ bursts back-to-back (separated each by a LIFS), where $N$ can be computed as
\begin{align}
 N := \max \left \{ 1, \rho \frac{ P^{(n)}_\idle}{1- P^{(n)}_\idle} \right\}, \label{eq:orla-aggregation}
\end{align}
and adjust the duration of the last burst accordingly. The corner case where there is no WiFi transmission whatsoever is discussed in \S\ref{sec:practical}.

\section{OLAA: Optimal Orthogonal LAA \\Access Strategy}\label{sec:olaa}

The ORLA transmission strategy introduced above maximizes the channel airtime an LBT WN can use while preserving the fairness constraint in Definition~\ref{def:fairness}. Since the amount of airtime held by an LBT WN when transmitting is fixed and equal to $T_{LBT}$, the ORLA policy also maximizes the throughput performance of asynchronous LBT WNs. However, this is not true for \emph{synchronous LBT WNs}, since the amount of airtime between the start of a transmission and the frame boundary this must synchronize with is effectively overhead, i.e., cannot be exploited for actual data communication, as illustrated in Fig.~\ref{fig:protocol}. Indeed, ORLA may prove highly inefficient in terms of throughput for synchronous LBT WNs, even providing less throughput than if following the WiFi specification, e.g., when the time until the next frame boundary is comparable to $T_{LBT}$. 


Recall that an LBT WN has to make a choice about whether to transmit or not every LIFS opportunity. If it decides to transmit at LIFS opportunity $\omega$, a synchronous LBT WN will send useful data for an amount of time equal to $Y^{(\omega)}= T_{LBT}-T_{res}$, where $T_{res}$ is the (random) time it takes between a LIFS opportunity and the closest frame boundary (see Fig.~\ref{fig:protocol}). Otherwise, it skips a \emph{round} and waits for the next LIFS opportunity $\omega+1$ (or round). We can regard this as an investment of time on each LIFS opportunity $\omega$ yielding a gross gain equal to $Y^{(\omega)}$ of useful channel time. Since this process is repeated over time, our goal is to design a policy $\Pi$ that maximizes the expected long-term \emph{rate of return}. 
Clearly, a naive policy that \emph{takes all} LIFS opportunities to transmit will not necessarily maximize the rate of return of a synchronous LBT WN, since the amount of time \emph{wasted} on waiting for frame boundaries may exceed the time invested in skipping LIFS opportunities. This resembles well-known problems in optimal stopping theory~\cite{ferguson}. 

Suppose now the above process is repeated $K$ times. Let $\{\omega_1, \cdots, \omega_k\}$ denote LIFS opportunities that have been taken (referred to as \emph{stopping times}), $Y^{(\omega_k)}$ the useful channel time (reward) obtained at opportunity $\omega_k$, and $\psi^{(\omega_k)}$ the time invested by doing so. Then, by the law of large numbers, 
$$
 \frac{\sum_{i=1}^K {Y^{(\omega_i)}}}{\sum_{i=1}^K \psi^{(\omega_i)}} \longrightarrow \frac{E[Y_\Pi]}{E[\psi_\Pi]}\, \text{a.s.}
$$
Then, we can cast the problem of maximizing the long-term average goodput of the synchronous LBT WN as a maximal-rate-of-return problem and use optimal stopping theory to solve it~\cite{ferguson}. In this way, we  need to characterize an optimal stopping rule $\Pi^\ddagger$ as
$$
\Pi^\ddagger := \argmax_{\Pi\in\boldsymbol{\Pi}} \frac{E[Y_\Pi]}{E[\psi_\Pi]}
$$
and the optimal LBT WN goodput as
$$
\lambda^\ddagger := \sup_{\Pi\in\boldsymbol{\Pi}} \frac{E[Y_\Pi]}{E[\psi_\Pi]}.
$$
Following \cite[Ch. 4, Th. 1]{ferguson}, we transform our maximal-rate-of-return problem into the following ordinary stopping rule problem:
\begin{align}
&\max E \left [ Y_\Pi - \lambda \left ( \sum_{i=1}^{\omega\in\Pi} T_{\text{slot}}K^{(i)} + T_{LBT} \right ) \right ],
\end{align}
where $K_i$ is the number of WiFi slots between LIFS opportunities $i$ and $i-1$. The intuition behind the above problem is that we invest $c_{\Pi}:=  \lambda\sum_{i=1}^\Pi T_{\text{slot}}K^{(i)}$ and gain $X_{\Pi}:= Y_{\Pi} - \lambda T_{LBT}$ in return, when we use stopping rule $\omega\in\Pi$. Then, according to \cite[Ch. 6, Th. 1]{ferguson}, the optimal rule $\Pi^\ddagger$ and the optimal throughput $\lambda^\ddagger$ is such that $$V^\ddagger(\lambda^\ddagger):= \sup_{\Pi\in\boldsymbol{\Pi}} E[Y_{\Pi} - \lambda \psi_{\Pi} ] = 0.$$ It can be shown then that our stopping rule is
$$
\Pi^\ddagger = \min\{ \omega \ge 1 \mid X_{\Pi} \ge V^\ddagger\}
$$
and that $V^\ddagger$ satisfies the optimality equation
$$
V^\ddagger = E\left [ \max \left \{ X^{(1)}, V^\ddagger \right \} \right ] - c^{(1)}.
$$
Given that $V^\ddagger(\lambda^\ddagger)=0$ and that $Y_{\Pi}$ is i.i.d, the above becomes
$$
E\left [ \max \left \{ Y-\lambda^\ddagger\left ( \frac{T_{\text{slot}}}{1-P_{\text{idle}}} + T_{LBT} \right ), -\lambda^\ddagger \frac{T_{\text{slot}}}{1-P_{\text{idle}}} \right \} \right ] = 0,
$$
since $E[K]=1/(1-P_{\text{idle}})$, and hence

\begin{align}
E\left [ Y - \lambda^\ddagger T_{LBT} \right ]^+ = \lambda^\ddagger\frac{T_{\text{slot}}}{1-P_{\text{idle}}}. \label{eq:fixed-point}
\end{align}
This is a fixed point equation that can be solved with iterative methods and the optimal rule is computed as 
\begin{align}
\Pi^\ddagger &= \min\left\{ \omega \ge 1 \mid Y_{\Pi}  \ge \lambda^\ddagger T_{LBT}\right\}\nonumber\\
&= \min\left \{ \omega \ge 1 \mid T_{res}  < T_{LBT}(1-  \lambda^\ddagger )\right\}.
\end{align}

In addition to finding the stopping policy that maximizes the rate of return, we must guarantee the conditions established in Lemma~\ref{th:barrho} or \ref{lem:bound}, i.e., the ratio of used LIFS opportunities must not exceed $$
 \pi(\bar\rho) := \bar\rho \frac{ P^{(n)}_\idle}{1- P^{(n)}_\idle }.
$$
In order to accommodate such constraint, let us first introduce the following lemma.

\begin{lemma}
\label{lem:uniform}
$T_{res}$ is uniformly distributed between $0$ and $T_{LBT}$. 
\end{lemma}
\begin{IEEEproof}
 We have a  time-slotted system $t=1, 2, \cdots$, each slot containing an idle, collision or successful WiFi event. 
 Let  $X:=\langle X_1, X_2, \cdots \rangle$ be a sequence of slots where $X_i$ is a Bernoulli trial with probability $\phi=1-P_{\text{idle}}$, i.e., a successful trial is a LIFS opportunity. This corresponds to a Bernoulli process $\mathcal{S}(t):=\sum_{i=1}^t X_i$ such the probability of having $K$ LIFS opportunities in $t$ slots   follows a binomial probability distribution, i.e., $P(\mathcal{S}(t)=K)=\mathcal{B}(t,K)$, and the number of slots between two LIFS opportunities follows a geometric distribution.  Let now $\mathcal{L}:=\langle \delta\cdot T_{LBT}\rangle$ be the sequence of slots containing an LTE-like frame boundary for $\delta=1,2,\cdots$ and denote 
 $T^{(\omega)}$ the (random) slot in which LIFS opportunity $\omega$ occurs. Then, since in order to be a LIFS opportunity, this must occur within the interval $[(\delta-1)T_{LBT}, \delta T_{LBT}),\ \delta=1,2,\cdots$, we need to compute the conditional distribution of $T^{(\omega)}$. As
 \begin{align}
 P\left ( T^{(\omega)} \le t \mid \mathcal{B}(T_{LBT}, 1)\right) &= \frac{P\left(T^{(\omega)} \le t ,\ \mathcal{B}(T_{LBT},1) \right)}{\mathcal{B}(T_{LBT},1)}\nonumber\\
 &=\frac{t \phi (1-\phi)^{t-1}(1-\phi)^{T_{LBT}-t}}{t \phi(1-\phi)^{T_{LBT}-1}}\nonumber\\
 &=\frac{t}{T_{LBT}}
 \end{align}
 is the CDF of an uniform distribution, $T_{res}=\mathcal{L}^{(\omega)}-T^{(\omega)}$ (note $\mathcal{L}$ is not a random process) is also uniformly distributed between $0$ and $T_{LBT}$.
\end{IEEEproof}

It is therefore sufficient to simply consider those LIFS opportunities closer than $ {\pi(\rho)}{T_{LBT}}$ slots to a frame boundary, that is,
\begin{align}
T_{res}  < {\pi(\bar\rho)}{T_{LBT}}.
\end{align}
This leads to the following theorem describing an optimal policy for synchronous LBT, which we name ``Orthogonal Licensed-Assisted Access'':
\begin{theorem}[\bf OLAA transmission policy]
 Consider a policy $\Pi^{\text{OLAA}}$ a synchronous LBT follows to initiate transmission according to the rule:
\begin{align}
\Pi^{\text{OLAA}}(\lambda, \rho) \!:=\! \left \{ \Pi \!\ge 1\! \mid\! T_{res}  \!\!<\! \min\!\Big( T_{LBT}(1\!\!-\!\!\lambda), {\pi( \rho)}{T_{LBT}}  \Big ) \!\right \}, \nonumber
\end{align}
where $T_{res}$ is the (random) time between a LIFS opportunity and the nearest frame boundary. Then, $\Pi^{\text{OLAA}}(\lambda^\ddagger, \bar\rho)$ is a throughput-optimal transmission policy $\Pi^\ddagger$ for synchronous LBT WNs, where $\lambda^\ddagger$ is computed by solving~\eqref{eq:fixed-point}, while $\pi(\bar \rho)$ and $\bar\rho$ are given by Lemma~\ref{lem:bound} and~\eqref{eq:optpi}, respectively.
\end{theorem}
\begin{proof}
 The proof follows directly from Lemma~\ref{lem:uniform}.
\end{proof}

\begin{figure*}[t!]
      \centering
      \subfigure[Asynchronous LBT, WiFi MPDU = 1500B, LBT Frame size = 1ms.]{
      \centering
      \includegraphics[trim = 0mm 0mm 0mm 0mm, width=0.45\linewidth ]{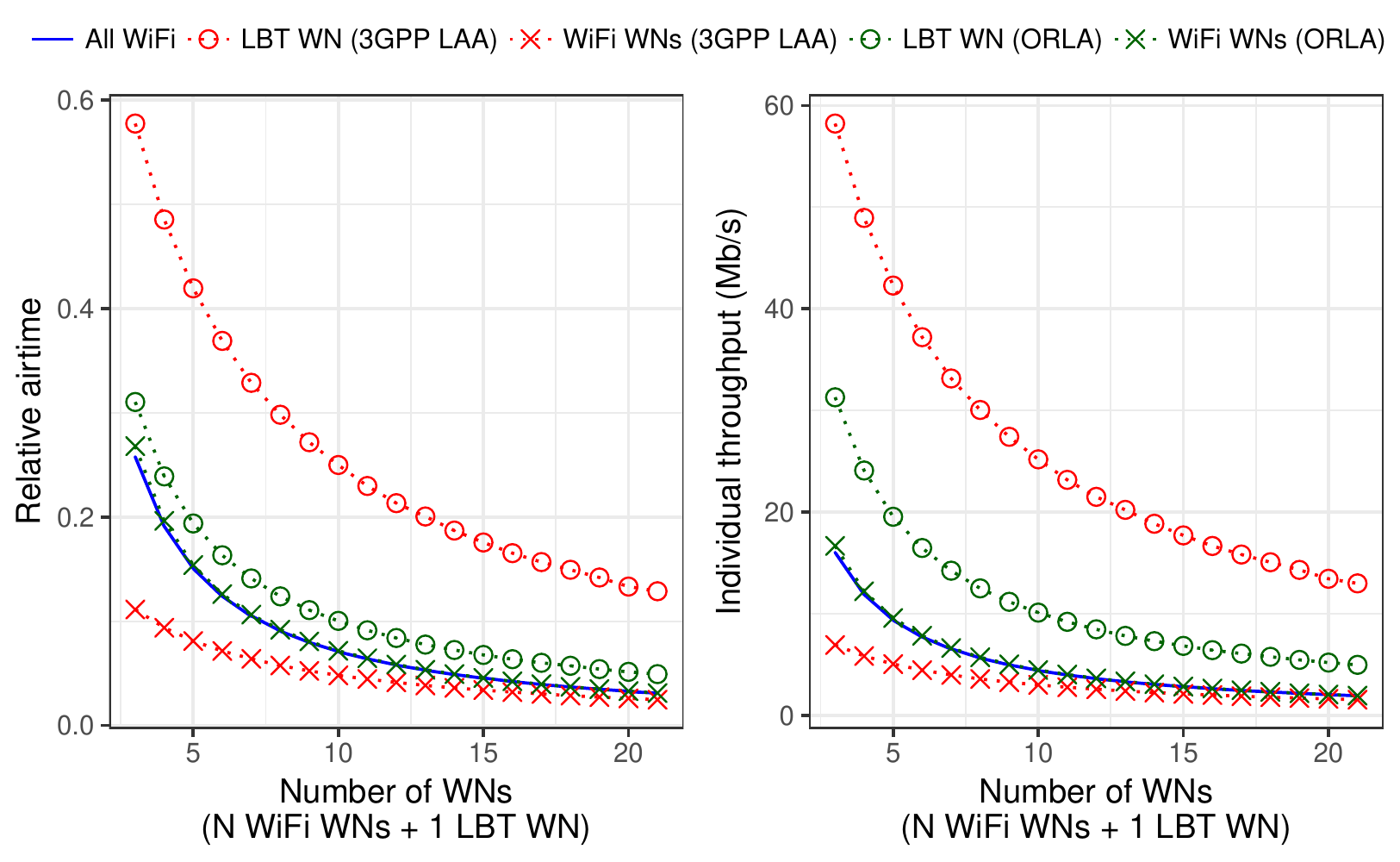}
      \label{fig:sims:sat-01:async}
      }
      \subfigure[Synchronous LBT, WiFi MPDU = 1500B, LBT Frame size = 1ms.]{
      \centering
      \includegraphics[trim = 0mm 0mm 0mm 0mm, clip, width=0.45\linewidth]{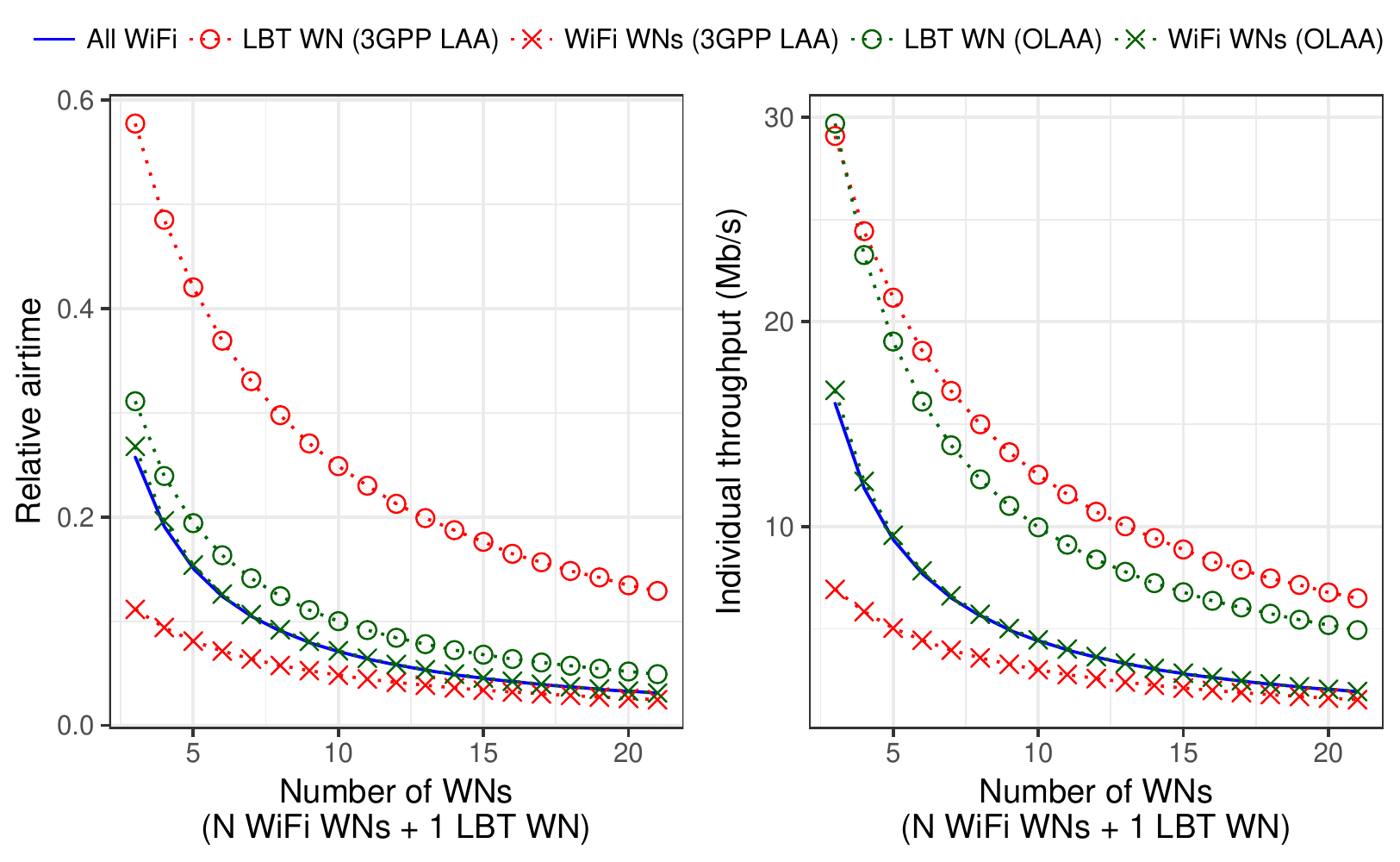}
      \label{fig:sims:sat-01:sync}
      }
      \subfigure[Asynchronous LBT, WiFi MPDU = 15000B, LBT Frame size = 1ms.]{
      \centering
      \includegraphics[trim = 0mm 0mm 0mm 0mm, width=0.45\linewidth ]{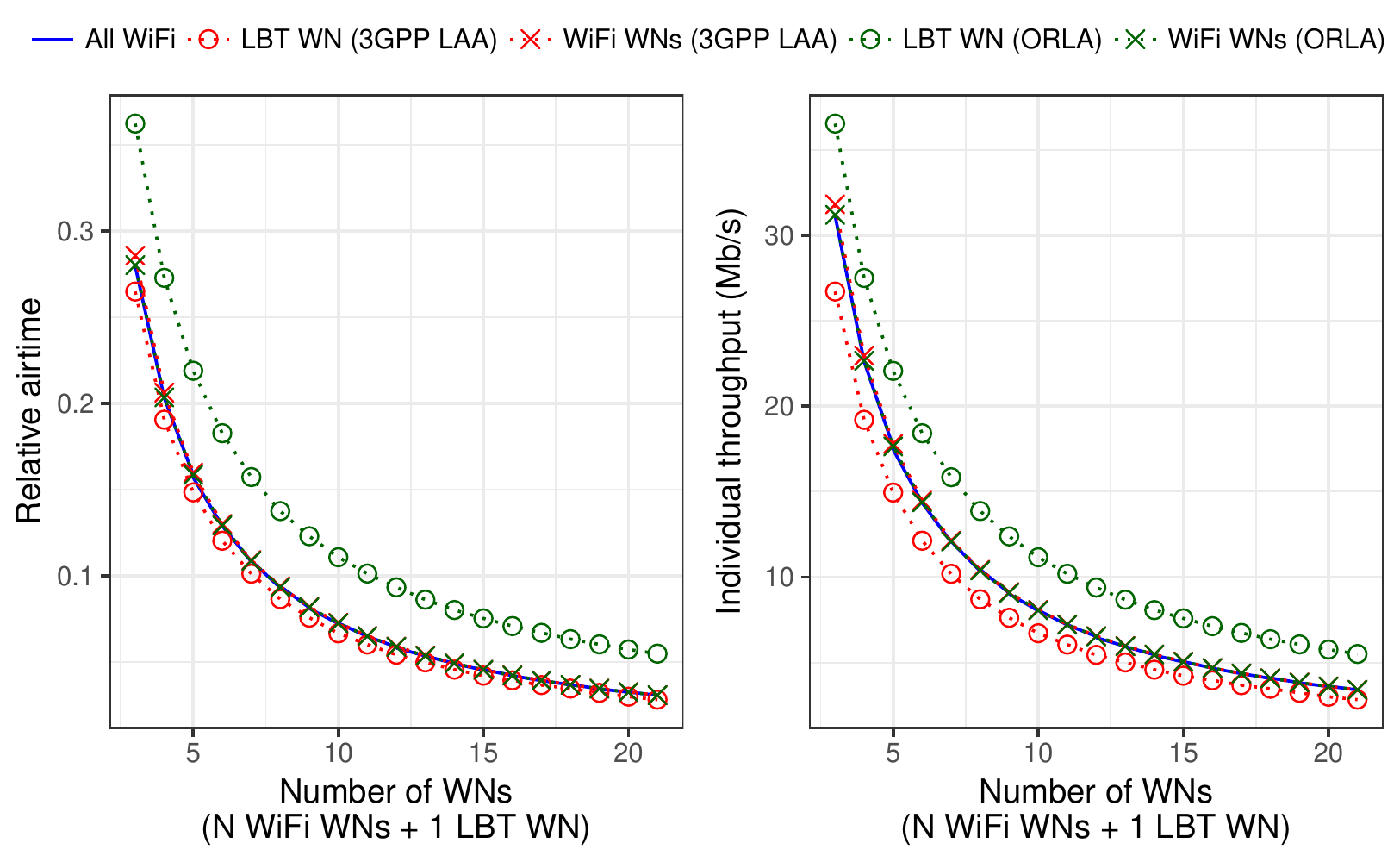}
      \label{fig:sims:sat-01:async:large-frame}
      }
      \subfigure[Synchronous LBT, WiFi MPDU = 15000B, LBT Frame size = 1ms.]{
      \centering
      \includegraphics[trim = 0mm 0mm 0mm 0mm, clip, width=0.45\linewidth]{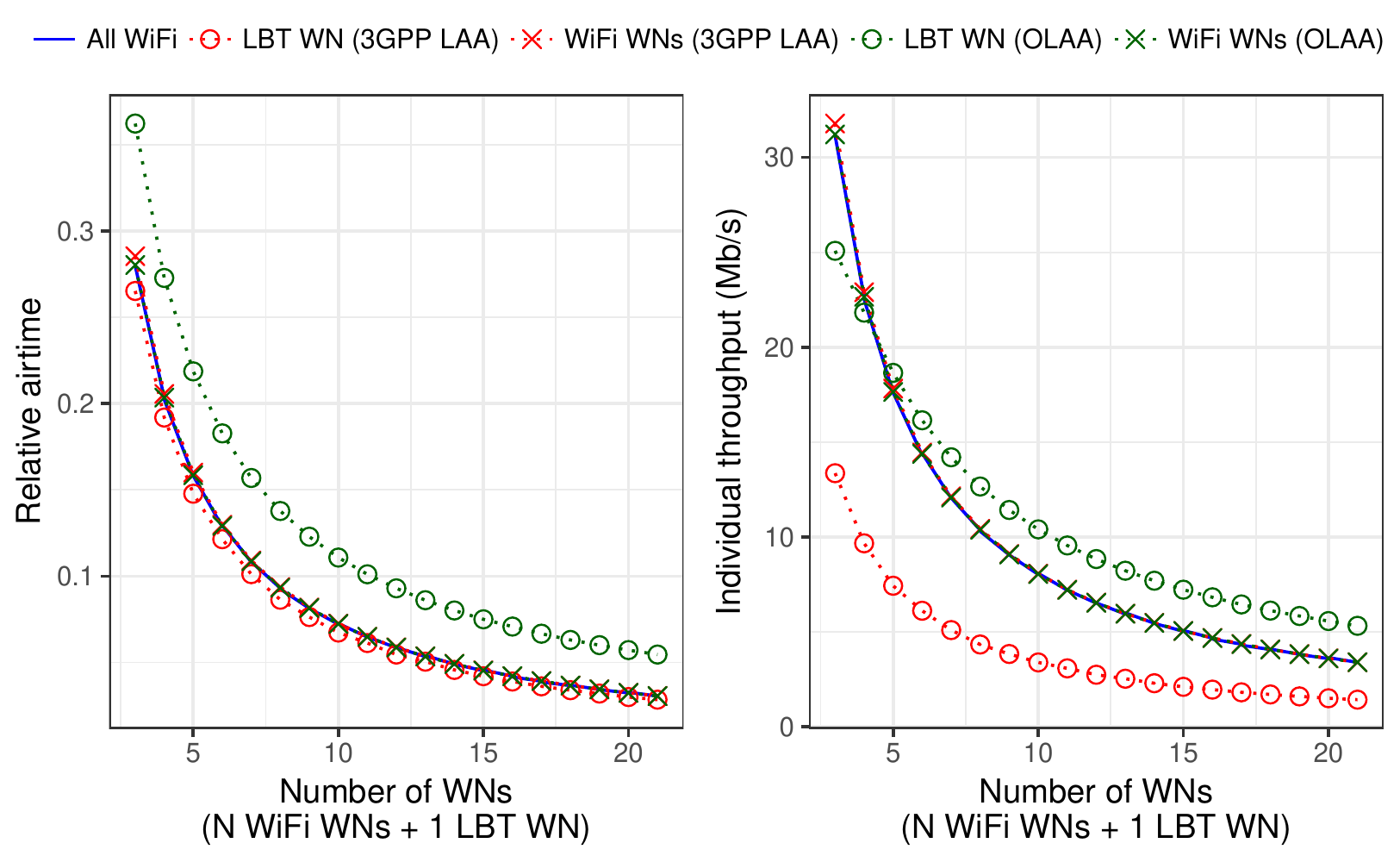}
      \label{fig:sims:sat-01:sync:large-frame}
      }
      \caption{Airtime and throughput performance of an LBT WN (triangles) operating with the proposed orthogonal coexistence mechanism with ORLA and OLAA, the benchmark 3GPP LAA, and the legacy WiFi stack, sharing the medium with a variable number of WiFi WNs. Performance of a background WiFi WN shown with crosses. Simulation results.}
      \label{fig:sims:sat-01}
      \vspace{-2mm}
\end{figure*}

\section{Simulation Results}
\label{sec:results}

In this section we undertake a comprehensive performance evaluation of the orthogonal LBT transmission policies we propose, i.e. ORLA and OLAA, by means of event-driven simulation. We will demonstrate our schemes attain superior throughput as compared to the de facto 3GPP LAA, while being substantially more fair to incumbent WiFi.
We consider coexistence with WiFi WNs that implement the IEEE 802.11ac specification \cite{IEEE80211ac}, with the parameters summarized in Table~\ref{tab:80211ac}. 
We examine the performance of both synchronous and asynchronous WNs. The latter employ the 3GPP's LAA protocol with the same contention parameters as WiFi WNs or our mechanism with ORLA policy. Unless otherwise stated, $CW_{\text{min}}=16$, $\bar{m}=4$, $T_{LBT} = 1$ms (i.e. LTE's Transmission Time Interval or TTI), and 64-QAM modulation is employed by both technologies. 

\begin{table}[b!]
\centering
\begin{tabular}{|c|c|}
\hline
Slot Duration ($\sigma$) & $9$~$\mu$s \\ \hline
DIFS & $34$~$\mu$s\\ \hline
SIFS & $16$~$\mu$s\\ \hline
PLCP Preamble and Headers Duration ($T_{\text{PLCP}}$) & 40~$\mu$s  \\ \hline
MPDU Delimiter Field ($L_{\text{del}}$) & 32 bits  \\ \hline
MAC Overhead ($L_{\text{mac-oh}}$) & 288 bits  \\ \hline
ACK Length ($L_{\text{ACK}}$) & 256 bits  \\ \hline
Data bit rate ($C$) & 130 Mb/s  \\ \hline
Control bit rate ($C_{\text{ctrl}}$) & 24 Mb/s  \\ \hline
\end{tabular}
\caption{IEEE 802.11ac~\cite{IEEE80211ac} parameters used for simulations.}\label{tab:80211ac}
\end{table}

\subsection{Variable Number of WiFi WNs}

\begin{figure*}[t!]
      \centering
      \subfigure[WiFi MPDU = 1500B, LBT Frame size = 1ms.]{
      \centering
      \includegraphics[trim = 0mm 0mm 0mm 0mm, width=0.45\linewidth ]{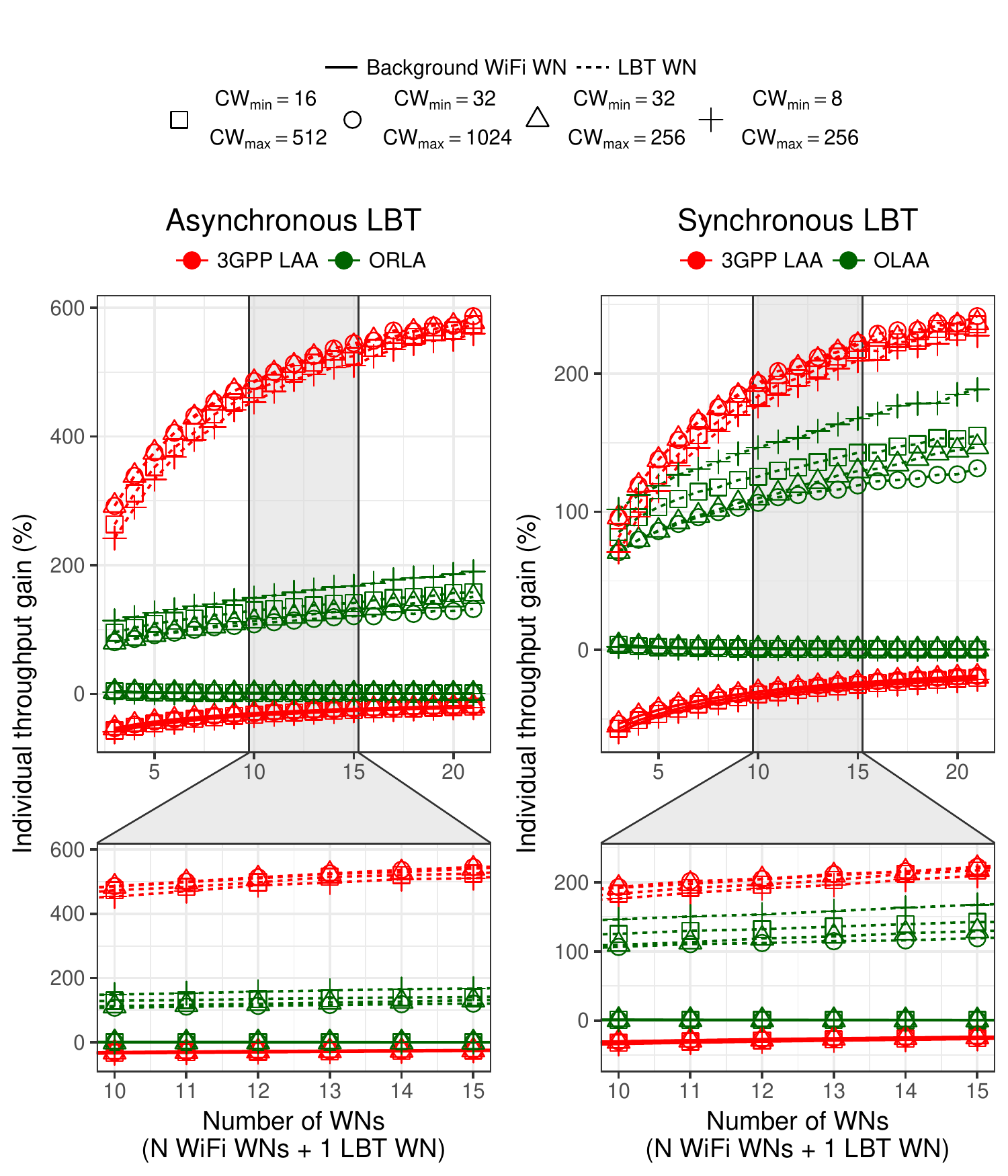}
      \label{fig:sims:sat-01:gain:small}
      }
      \subfigure[WiFi MPDU = 15000B, LBT Frame size = 1ms.]{
      \centering
      \includegraphics[trim = 0mm 0mm 0mm 0mm, clip, width=0.45\linewidth]{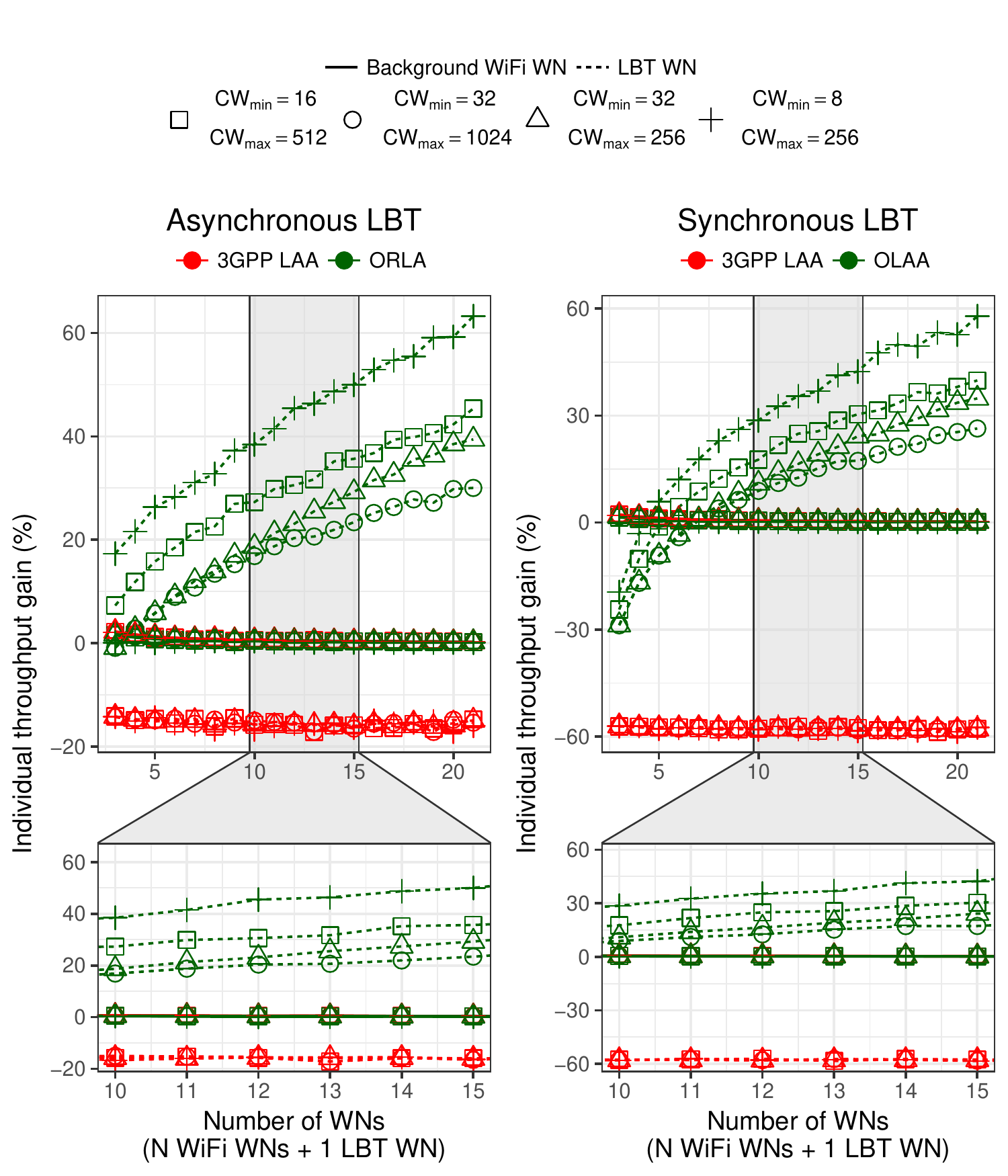}
      \label{fig:sims:sat-01:gain:large-frame}
      }
      \caption{Individual throughput gain different LBT approaches achieve with respect to legacy WiFi for a variable number of background WiFi WNs, different contention settings, and different MPDU sizes. LBT WNs employ (a)synchronous operation. Simulation results. }
      \label{fig:sims:sat-01:gain}
      \vspace{-2mm}
\end{figure*}

We first investigate the airtime and individual throughput performance of an LBT WN operating with our LBT access mechanism with ORLA and OLAA policies, the benchmark 3GPP LAA, and the legacy WiFi protocol, as we vary the number of (background) WiFi WNs sharing the channel. We consider backlogged background WiFi transmitters, first working with $f_{\text{agg}}=1$ and payload $B=1500$B, and subsequently aggregating 10 packets (i.e., $f_{\text{agg}}=10$) and sending $B=15000$B upon each attempt. In these experiments, the coexisting LTE WN works with $T_{LBT} = 1$ms. Our measurement results are presented in Fig.~\ref{fig:sims:sat-01}.

It is important to observe first the behavior of 3GPP LAA and background WiFi WNs when the latter transmits bursts of $B=1500$B (the red points Figs~\ref{fig:sims:sat-01:async}--\ref{fig:sims:sat-01:sync}). Clearly LAA consumes almost 6$\times$ more airtime when the number of WiFi contenders is small, which results in up to a 2-fold reduction in the throughput of a WiFi WN (blue line vs. red crosses). This violates the first coexistence criterion, i.e., not harming the performance of incumbent WiFi. In contrast with our approach, both ORLA and OLAA, safeguard WiFi throughput (overlapping green crosses and blue line), while the LBT WN consumes additional airtime more wisely, almost doubling MAC throughput efficiency (which is the second coexistence criterion) without negatively impacting on WiFi (green circles). As such, 3GPP LAA attains more throughput as compared to ORLA though at the mentioned price (Fig.~\ref{fig:sims:sat-01:async}), though in synchronous mode of operation where frame alignment is required, the throughput performance of OLAA and 3GPP LAA are comparable, but ORLA consumes half the airtime thereby giving more opportunities to WiFi and ensures harmless operation (Fig.~\ref{fig:sims:sat-01:sync}).

Interestingly, when WiFi WNs transmit large payloads, i.e. MPDU=15000B, 3GPP LAA does not harm WiFi performance, but neither does it attain superior throughput efficiency (Figs.~\ref{fig:sims:sat-01:async:large-frame}--\ref{fig:sims:sat-01:sync:large-frame}). In fact, although the relative airtime of the two technologies is comparable, 3GPP LAA exhibits inferior throughput, even more so when operating synchronously, in which case throughput performance of LBT can be even less than half of that of WiFi (Fig.~\ref{fig:sims:sat-01:sync:large-frame}). Unlike the 3GPP benchmark, ORLA and OLAA do consume more airtime, though without affecting WiFi performance (again green crosses overlapping with blue line). This leads to a constant throughput gain, irrespective of the number of contenders. Importantly, under synchronous operation OLAA achieves twice the throughput of the 3GPP benchmark. These results suggest that \textbf{IEEE 802.11ac may prove more efficient than 3GPP LAA, whereas our proposed coexistence schemes bring up to 100\% throughput gain without harming WiFi.}

\subsection{Effects of Contention Parameters}
\begin{figure*}[t!]
      \centering
      \subfigure[Asynchronous LBT.]{
      \centering
      \includegraphics[trim = 0mm 0mm 0mm 0mm, width=0.44\linewidth ]{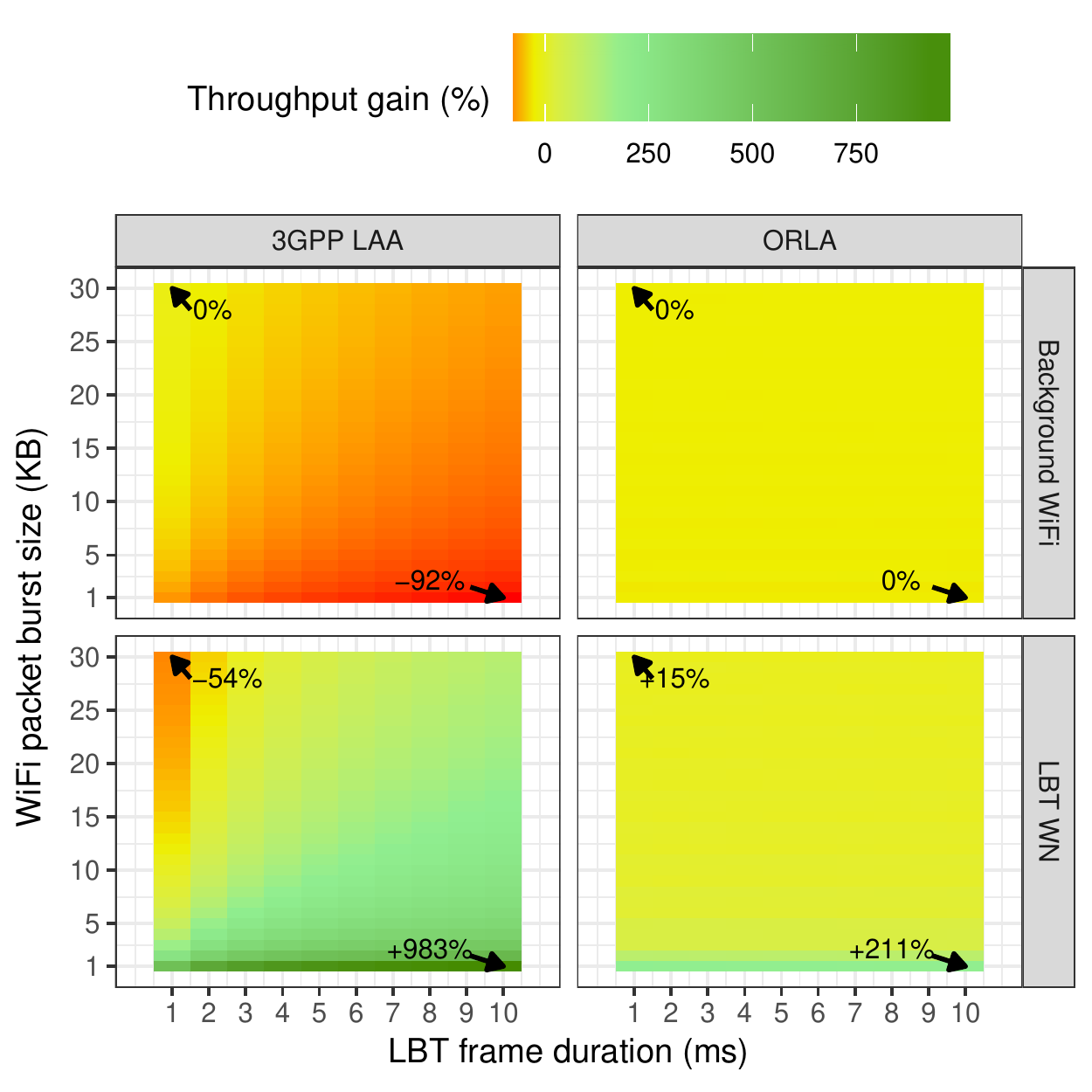}
      \label{fig:sims:sat-01:gain:all:async}
      }
      \subfigure[Synchronous LBT.]{
      \centering
      \includegraphics[trim = 0mm 0mm 0mm 0mm, clip, width=0.44\linewidth]{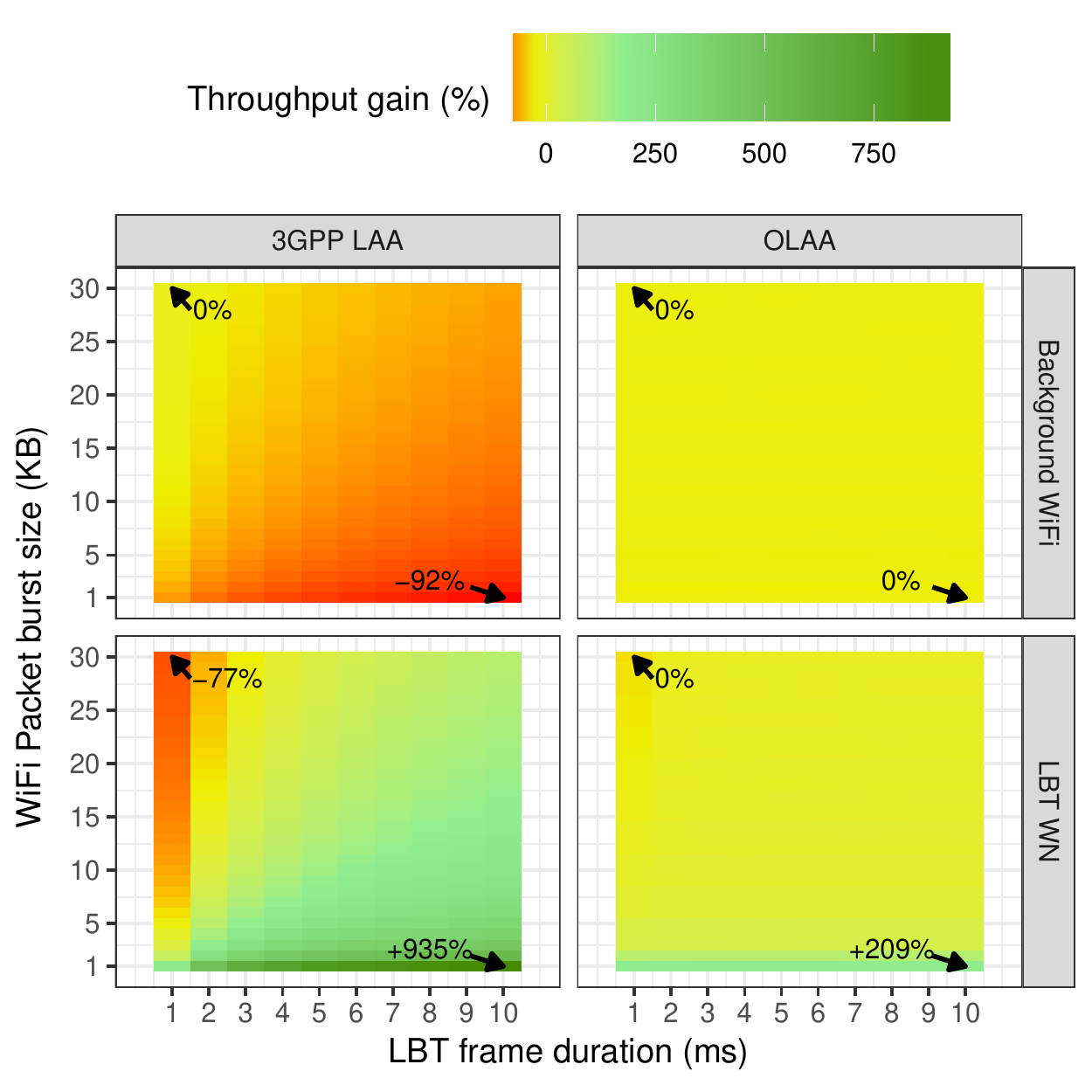}
      \label{fig:sims:sat-01:gain:all:sync}
      }
      \caption{Individual throughput gain with respect to legacy WiFi when a LBT WN shares the channel with 5 background WiFi WNs. Different LBT frame and WiFi burst sizes employed. Simulation results.}
      \label{fig:sims:sat-01:gain:all}
      \vspace{-4mm}
\end{figure*}

Although the optimal configuration of WiFi WNs is outside the scope of our work, since the 802.11 standards allows for adapting the contention settings ($CW_{\min}$ and $CW_{\max}$), we are interested in quantifying the performance gains LBT may obtain under different configurations, and irrespectively any losses incurred onto WiFi. To this end, we consider both synchronous and asynchronous LBT operation, fixed LBT frame size (1ms), different WiFi burst sizes (1500 and 15000B), and 4 CW configurations with different aggressiveness, namely (16,512), (32,1024), (32,256), and (8,256), while we vary again the number of WiFi WNs. The results of these experiments are illustrated in Fig.~\ref{fig:sims:sat-01:gain}, where we plot the individual throughput gains with respect to a scenario where the LBT WN were an additional WiFi WN.

We note that when the background WiFi stations transmit small payloads and the LBT WN operates with the 3GPP LAA scheme, the difference between contention configurations are subtle. However, the remarkable 3GPP LAA throughput gains are at the expense of WiFi losses (negative gain), as shown in Fig.~\ref{fig:sims:sat-01:gain:small}. In contrast, our orthogonal coexistence approach imposes no penalty on WiFi (observe the constant lines at gain equal to 0\%), while we achieve throughput gains almost up to 200\% with both ORLA and OLAA policies. These gains are slightly more prominent when WiFi contends aggressively (i.e. $CW_{\min}=8$), which we attribute to more LIFS opportunities higher WiFi attempt rates create. The small differences between contention window parameters become clearer in the bottom of the figure, where we zoom in between 10 and 15 WiFi WNs. 

Our previous finding about the questionable efficiency of 3GPP LAA in scenarios where WiFi WNs transmit 15000B payloads is further confirmed in Fig.~\ref{fig:sims:sat-01:gain:large-frame}, where we observe that LAA incurs no penalty to WiFi WNs (which virtually lie on a gain equal to 0\% irrespective of the number of WNs) but 3GPP LAA has between 15 and 60\% of throughput loss over the case where it followed the WiFi protocol of the background WNs (this become clearer at the zoomed area, at the bottom part of the figure). In contrast, ORLA and OLAA do not affect WiFi performance in neither asynchronous or synchronous settings. We do note that when the number of WiFi WNs is small (mostly less than 5), OLAA does not attain throughput gains because the LBT frame size is fixed and no aggregation is allowed (we will assess optimal ORLA/OLAA aggregation, as suggested in Eq.\eqref{eq:orla-aggregation}, later on). However, depending on contention setting, \textbf{individual LBT performance grows to as much as 60\% under high contention levels, with both orthogonal coexistence policies we propose.}

\subsection{Impact of LBT Frame and WiFi Burst Durations}

Next, we provide further insight into the impact of different WiFi burst sizes, as well as $T_\text{LBT}$ settings, i.e., the duration for which an LBT holds the channel when transmitting. We are particularly interested in understanding the impact LBT may have on WiFi when working with large frame sizes, since by default LTE operates with 10ms frames, while the ETSI specification allows transmission in the unlicensed band for up to 8ms. For this purpose, we consider a scenario with 5 backlogged WiFi WNs and the LBT WN operating with 3GPP LAA scheme and with the orthogonal coexistence mechanism we propose, respectively. We investigate again both asynchronous and synchronous operation, whereby our solution employs ORLA and OLAA policies, respectively. The results of this experiment are shown in Fig.~\ref{fig:sims:sat-01:gain:all}, where we plot as heatmaps the throughput gains of both LBT (with each approach) and background WiFi.

Under asynchronous operation, it can be seen that the 3GPP LAA may attain as much as 983\% throughput gains when working with 10ms frames. However, this cuts off WiFi transmissions almost completely (92\%) throughput loss. Conversely, if WiFi performance is to be preserved, 3GPP LAA has minimal gains as compared to using legacy WiFi, and may even be 54\% less efficient if the background WiFi WNs employ long bursts (Fig.~\ref{fig:sims:sat-01:gain:all:async}). This effect is further exacerbated in the case of synchronous LBTs, as illustrated in Fig.~\ref{fig:sims:sat-01:gain:all:sync}.

In contrast \textbf{both ORLA and OLAA achieve more than 200\% throughput gains when operating with long LBT frames, without negatively impacting on WiFi}. Indeed the performance of the incumbent remains unaffected, irrespective of the LBT frame/WiFi burst settings. Further, ORLA achieves improvements even when sending 1ms frames.

\subsection{Non-saturation Conditions}

We now turn attention to circumstances where incumbent WiFi WNs have limited offered load (i.e., non-saturation). We expect LBT to effectively exploit the additional airtime available in light load regimes, without harming WiFi performance. To this end, we consider 5 WiFi WNs transmitting 1500B PDUs and increase their offer load, relative to the load that saturates the network, from 10\% (light load) to 100\% (saturation conditions).
We study the performance of both LBT and background WiFi when the LBT frame duration is 1 and 10ms, respectively; the LBT WN operates again with either 3GPP LAA or the proposed orthogonal coexistence scheme with ORLA (asynchronous) or OLAA (synchronous) policy, respectively. Note however, that aggregation level in this case is automatically adjusted for both ORLA and OLAA, as explained in Eq.~\eqref{eq:orla-aggregation}. 

\begin{figure}[t!]
      \centering
      \includegraphics[trim = 0mm 0mm 0mm 0mm, clip, width=\linewidth]{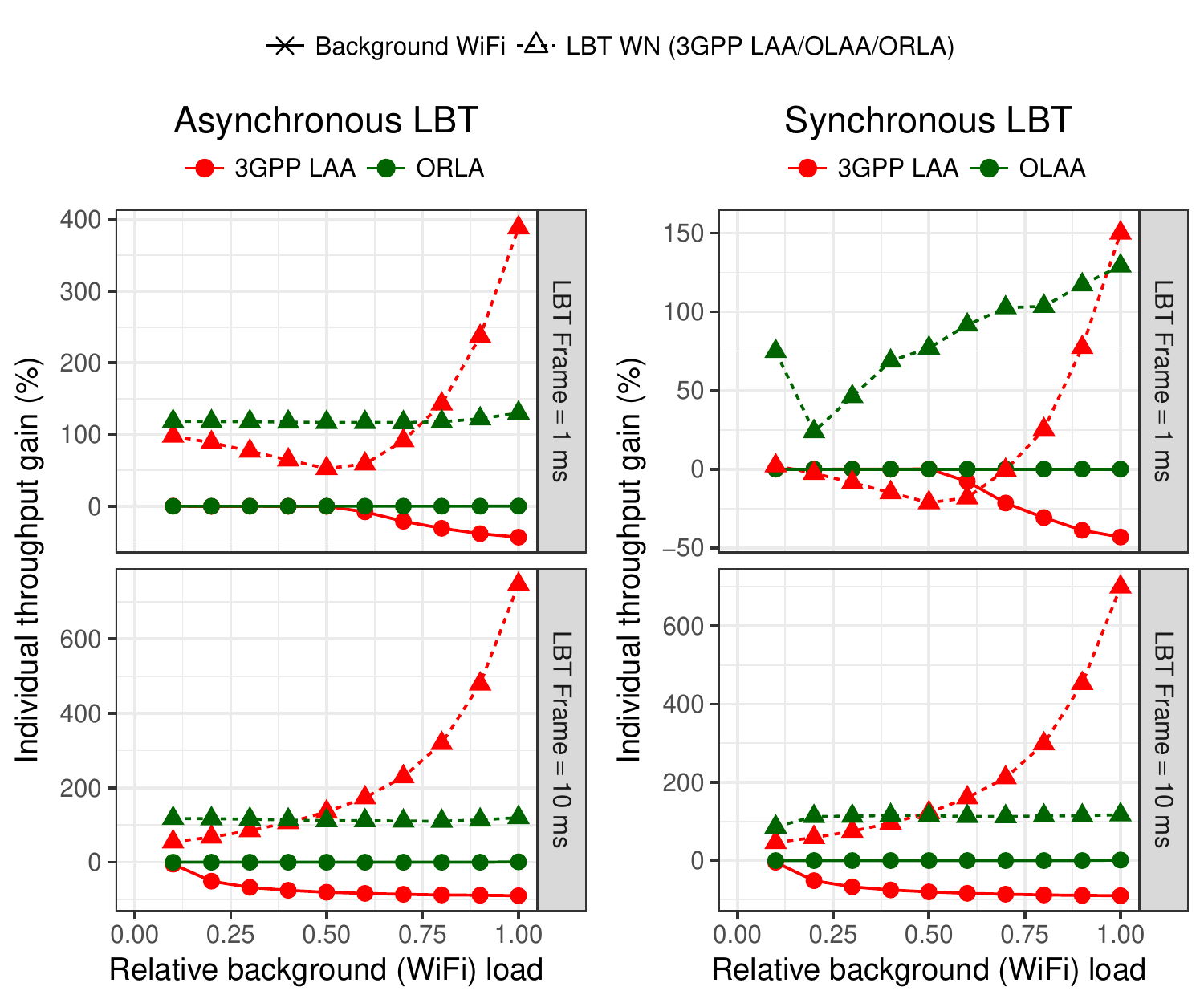}
      \caption{Individual throughput gain with respect to legacy WiFi when an LBT WN shares the medium with 5 background WiFi WNs, whose offered load increases. Simulation results. }
      \label{fig:sims:nonsat-01}
      \vspace{-4mm}
\end{figure}

The obtained results are plotted in Fig.~\ref{fig:sims:nonsat-01}. It can be seen that 3GPP LAA negatively impacts on WiFi even when lightly
loaded. This is more obvious when the frame duration is long (10ms), while we note that when the $T_\text{LBT}=1$ms, 3GPP LAA leaves WiFi unaffected but exhibits decreasing performance up to the point where the relative load is precisely 50\% (observe the LBT minimum), following which the LBT gain grows at the expense of WiFi. \textbf{The proposed transmission policies do not affect non-saturated WiFi WNs. ORLA provides steady throughput gains above 100\% up to the point where the WLAN saturates, OLAA's performance grows with WiFi activity level, again exceeding 100\% improvements.}

\subsection{Heterogeneous Bitrates}

We conclude the evaluation by investigating the performance of the proposed scheme and policies in a multi-rate scenario where 5 backlogged WiFi WNs transmit at different bitrates in response to dissimilar channel conditions. Specifically, upon accessing the channel, each transmits 1500B at the following rates $\{156, 130, 78, 39, 13\}$ Mb/s, respectively. This will yield longer slot durations whenever a slower station transmits, leaving less time available for both LBT and other WiFi contenders. In this scenario, the LBT WN operates with a frame of 1ms and transmits using 64-QAM (MCS level 6). We compare again against the 3GPP LAA benchmark in both asynchronous and synchronous settings. 

As seen in Fig.~\ref{fig:sims:multirate-01}, with asynchronous operation the 3GPP LAA attains remarkable gains as compared to legacy WiFi (500\%), but has a negative impact on the WiFi contenders. In contrast, ORLA ensures harmless coexistence, while still providing nearly 200\% throughput gains. When the LBT transmissions align to frame boundaries (synchronous LBT), the 3GPP LAA's gains are less impressive, while WiFi is more severely affected. The results obtained confirm the OLAA policy does not inflict penalties onto WiFi also in this case, while still achieving 150\% performance gains.

\begin{figure}[t!]
      \centering
      \includegraphics[trim = 0mm 0mm 0mm 0mm, clip, width=\linewidth]{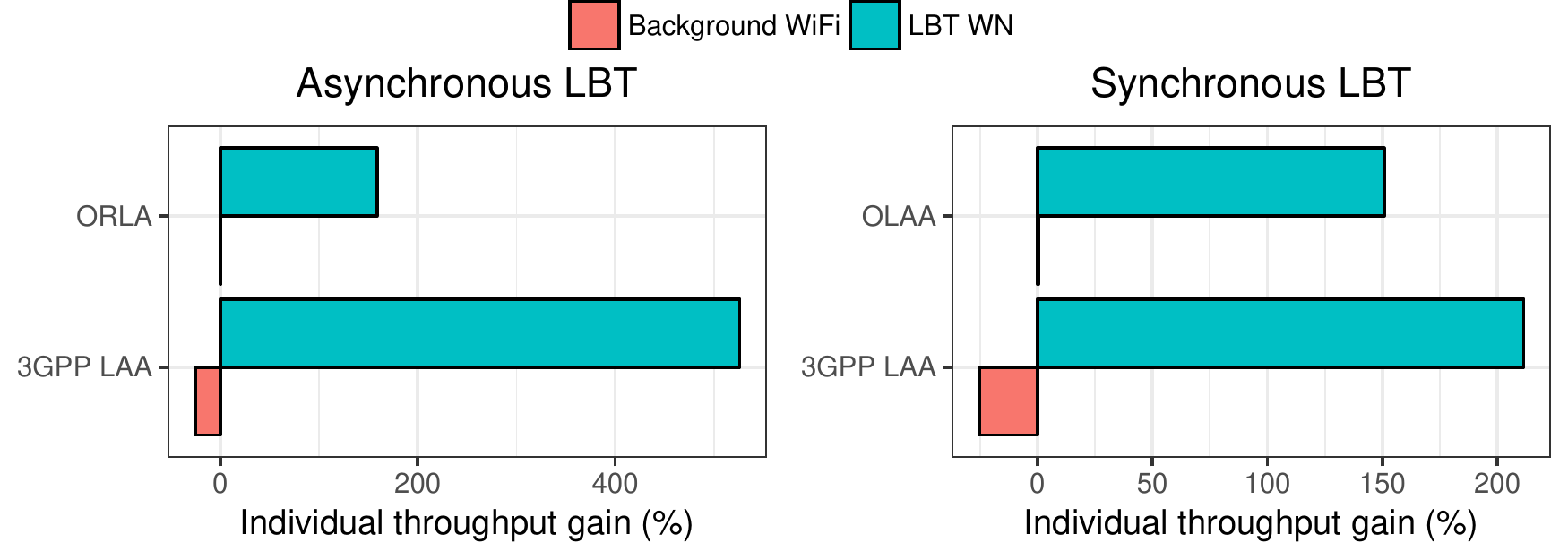}	
      \caption{Individual throughput gain with respect to legacy WiFi in a multi-rate scenarios where an LBT WN share the channel with 5 saturated WiFi WNs. WiFi's MCS $\in \{156, 130, 78, 39, 13\}$ Mb/s, LBT MCS level = 6. WiFi MPDU = 1500B, LBT Frame size = 1ms. Simulation results.}
      \vspace{-2mm}
      \label{fig:sims:multirate-01}
\end{figure}

\section{Practical Considerations}\label{sec:practical}


In this section we discuss the implications and validity of some of the assumptions made in our theoretical analysis, and give a set of technical guidelines for practical implementation of the ORLA and OLAA coexistence mechanisms. Recall that our objective is to provide superior MAC efficiency as compared to 3GPP's LAA, while protecting native WiFi systems. 

\subsection{PHY layer aspects}
Our focus is on the MAC layer and in our modeling we have made number of assumptions regarding the physical layer, primarily for tractability. We revisit these and explain how they could be captured more precisely in practice. In particular, 
\begin{enumerate}[label=($\arabic*$)]
 \item \emph{Channel errors:} We assumed channel errors are completely handled by appropriate use of Modulation and Coding Schemes and as such packet loss occurs only due to collisions. Our analysis can however be readily extended to account for noise/fading induced channel losses, which directly impact the average slot duration. In \cite{Patras:2015:ADHOCNETS} we showed that the average slot duration can be inferred by only measuring the duration of successfully received frames, which an LBT WN can easily perform via packet sniffing. Hence, accommodating channel errors makes use of available measurements and does not require explicitly quantifying the loss rate.
 \item \emph{Hidden terminals:} We note that hidden WNs could initiate a transmission even when another WN (WiFi or LBT) has already acquired the channel, therefore the attempt probabilities are coupled with the dynamics of the transmissions not directly overheard. While the goal of our work is not to address such WiFi-related issues, which are already well studied in the literature (see e.g. \cite{giustiniano07}), the coexistence mechanism we propose aims to avoid harming incumbent WiFi. For this purpose, a non-WiFi LBT WN can minimize the amount of \emph{additional} hidden terminal issues by transmitting (at a sufficiently high power) a CTS-to-self every time a LIFS opportunity is taken. 
 \item \emph{Capture effect:} Our modeling assumes that concurrent transmissions cannot be decoded at the intended receivers, i.e., they result in collisions. In practice, however, a receiver could still decode a `strong' packet, if the receive power difference between multiple transmissions is sufficiently large. This may cause unfairness among WiFi WNs~\cite{patras2014mitigating}, which our design does not aim to mitigate. However, the proposed coexistence mechanism neither alters the behavior of incumbent WiFi, nor induces additional unfairness due to capture effect issues. Precisely, WiFi remains agnostic about the orthogonal access procedure we introduce. 
\end{enumerate}

\subsection{MAC layer aspects}
Regarding the MAC behavior modeling underpinning the design of our LBT scheme, we did not explicitly consider multi-operator scenarios, absence of WiFi activity, details concerning the computation of the average duration of slots containing collisions in multi-rate settings, and the impact on delay performance. We discuss these aspects next.
\begin{enumerate}[label=($\arabic*$)]
  \item \emph{Multiple non-WiFi LBT WNs:} Assuming a channel is mostly used by a single non-WiFi LBT WN, as in our model, is reasonable, since we expect that LTE operators have financial incentives to coordinate offline on exploiting spatial and spectral diversity. On the other hand, in case multiple non-WiFi LBT WNs are forced to contend for unlicensed spectrum, they can either coordinate transmissions via cross-carrier downlink control information (DCI) \cite{ahmadi:2014} or could distributively contend for non-WiFi LBT orthogonal channel time via CSMA methods.
  \item \emph{Absence of WiFi transmissions:} Our LBT design is built upon a seamless symbiosis between WiFi and non-WiFi LBT transmissions. Lack of WiFi activity in a network will effectively lead to no transmission opportunities for a LBT transmitter following our scheme. However, as we  employ carrier sensing, such circumstances can be detected and we argue in favor of dynamic switching to 3GPP's LAA in this corner case, especially since channel access would not be subject to collisions.
   \item \emph{Estimating collision durations:} The orthogonal transmission policies we propose depend on accurate knowledge of the average slot duration, and therefore the duration of slots containing collisions. Estimating this was long thought to be difficult, but recent results demonstrate analytically and practically that average collision durations can be accurately computed by simply inspecting correctly received packets~\cite{Patras:2015:ADHOCNETS}.   
   \item \emph{Impact on delay.} An LBT WN may impact on the delay experienced by background WiFi WNs when accessing the medium in two ways: ($i$)~by causing collisions with WiFi transmissions, and ($ii$) by continuous airtime consumption. By design our coexistence mechanism and transmission policies, in contrast to LAA, incur minimal delay penalties on WiFi, since ($i$)~our orthogonal access paradigm does not cause collisions, and ($ii$)~the LBT frame size can configured as equal (or smaller) to the amount of time taken by a WiFi transmission, with no cost in performance (see Fig.~\ref{fig:sims:sat-01:gain:all}).
\end{enumerate}

\subsection{Practical implementation}
In the remainder we discuss how the coexistence mechanism and the optimal transmission policies we propose could be implemented with off-the-shelf wireless hardware. Recall our design relies on the assumption that LBT WNs are equipped with native LTE transceivers and additional IEEE 802.11 network cards. The latter are inexpensive, already provide access to a number of channel state parameters required by our policies (e.g. duration of idle/busy periods via open source firmware~\cite{openfwwf}) and have been used to develop custom features (e.g. precise control of TX engine timings~\cite{Facchi2016ComCom}). For our design, key functionality can be implemented as follows.

\begin{enumerate}[label=($\arabic*$)]
  \item \emph{Estimating the number of WiFi WNs:} By employing passive packet sniffing on the 802.11 interface, the LBT WN can employ a range of practical methods to estimate the number of active WiFi contenders, including Kalman filtering~\cite{1208922} and Bayesian estimation~\cite{1415933}.
  \item \emph{Estimating transmission probabilities ($\tau$):} All the parameters required by our optimal transmission policies can be derived from the transmission attempt probabilities of WiFi WNs, $\tau$, as we explain in \S\ref{sec:orla}. Under saturation conditions, $\tau$ can be computed by solving a non-linear system of equations which depend on contention parameters $\text{CW}_{\text{min}}$ and $m$, and conditional collision probabilities $p$ (see \S\ref{sec:setup}). The former can be directly extracted from overheard \emph{beacon} frames sent periodically APs. The latter can be estimated by inspecting the retry flag of  overheard WiFi packets~\cite{serrano2012control}. In case of non-saturation, an LBT WN also needs to estimate the probability $q$  that a new packet arrives in a uniform slot time $T_{\text{slot}}$. If each WN managed sufficiently large number of heterogeneous flows, by the Palm--Khintchine theorem~\cite{palm-khinchine}, the resulting arrival process will be Poisson and $q$ can be related to the average load as $\lambda = -\log(1-q)/T_{\text{slot}}$. Although this assumption may not always hold, it is commonly used for testing purposes, including for LTE/WiFi coexistence evaluation~\cite{qualcomm-lteu}.
  \item \emph{Detecting LIFS opportunities:} To avoid collisions, an LBT WN has to detect if a channel has been idle for LIFS duration. This can be achieved by inspecting channel state registers of 802.11 cards and signaling such event through the device's shared memory~\cite{Facchi2016ComCom}, which the LTE protocol stack can periodically inspect. In case of synchronous LBT operation, the 802.11 interface can also be instructed by writing into the shared memory, to immediately generate a CTS-to-self frame that reserves the channel for the desired duration. 
\end{enumerate}

\section{Related Work}\label{sec:related}

The coexistence between IEEE 802.11-based WiFi networks and other wireless technologies in unlicensed spectrum has received much interest from the research and industry communities. Coexistence mechanisms have been repeatedly proposed to handle joint operation of WiFi and Bluetooth~\cite{friedman2013power}, WiFi and Zigbee~\cite{5672592, zhang2011enabling, tytgat2012avoiding}, WiFi and WiMaX \cite{andreev2011ieee}, and more recently between WiFi and LTE systems operating in unlicensed bands~\cite{6655388,6554707, babaei2014impact, cavalcante2013performance,gomez2016srslte,6328367,sagari,syrjala2015coexistence,sadek2015extending,7218451,guan4cu,laa-multicarrier,7247522,canocoexistenceton,capretti2016,Chai:2016}. 

The work in \cite{6655388} is one of the first to addressing LTE/WiFi coexistence, proposing a simple scheme based on Almost Blank Subframes (ABS) to support coexistence in a 900~MHz (TV) band. 
The performance of WiFi in the presence of LTE was analyzed numerically~\cite{cavalcante2013performance} and by means of simulation~\cite{6554707, babaei2014impact}, though the LTE access model in these studies is not based on listen-before-talk and therefore it is not compliant with ETSI regulation.  Experimental results of CSAT (Carrier Sensing Adaptive Transmission) based access over software-defined platforms were reported in~\cite{gomez2016srslte}, the impact of U-LTE CSAT duty cycling patterns were quantified in~\cite{capretti2016}, and WiFi signals were embedded into LTE in~\cite{Chai:2016} to improve coexistence.

Ratasuk \emph{et al.} studied an LBT-enabled LTE access mechanism very similar to what later was specified as Licensed-Assisted Access (LAA)~\cite{6328367}.
3GPP presented early simulation results with LBT-based LAA, claiming that in some scenarios an LAA node can be configured such that the performance of WiFi users would not be affected more than if another WiFi station were added to the network~\cite{tr36.889}. 
However, as we demonstrated through our evaluation, whether legacy LAA meets this fairness criteria while improving spectral efficiency (as our work does) remains questionable.

In \cite{sagari}, Sagari introduced an inter-network coordination architecture to solve the bandwidth and channel selection problem when LTE systems access WiFi bands.  
Coexistence was modeled as a classic opportunistic spectrum access problem in~\cite{syrjala2015coexistence}, whereby LTE and WiFi users are considered secondary primary users respectively, operating with cognitive full-duplex transceivers that follow a non ETSI-compliant access protocol. 

Sadek \emph{et al.} focused on a similar fairness criterion as the one we consider, i.e. in a network with a set of WiFi stations, if an arbitrary number of WiFi WNs are replaced with U-LTE nodes, the performance of the remaining WiFi stations should be comparable (or better) than in the all-WiFi scenario~\cite{sadek2015extending}. The authors undertook a simple analysis and simulation-based evaluation of a duty cycle-based approach (CSAT), which is not ETSI-compliant.
More recently, Guan \emph{et al.} analyzed the problem of channel selection and carrier aggregation when using U-LTE CSAT-like opportunistic spectrum access~\cite{guan4cu}, while Liu \emph{et al} proposed an LBT-based mechanism for LAA  multi-carrier operation~\cite{laa-multicarrier}.

In \cite{7218451}, Yun  \emph{et al.} made a radically different proposal where both LTE and WiFi signals are transmitted together and a novel decoding method is designed to decode the overlapped transmissions. However, this scheme requires a redesign of the physical layer for both LTE and WiFi, which renders practical implementation cumbersome. 

The work of Cano \emph{et al.} in \cite{7247522} and \cite{canocoexistenceton} proposed coexistence mechanisms based on both CSAT duty-cycling and LBT, which achieve a proportionally fair rate allocation across LTE/WiFi nodes. We note that our fairness criteria, in line with 3GPP's, fundamentally differs from this approach, since WiFi systems do not naturally lean towards proportional fairness.

The interest in U-LTE shown by industry is unquestionable. While Nokia, Qualcomm, and Huawei released white papers showing promising results~\cite{nokia, qualcomm,huawei}, the details of the access schemes and simulation models used are not public. 
Intel presented a performance analysis of the coexistence between both technologies from a system level perspective, using a large scale model based on stochastic geometry where both U-LTE and WiFi employ CSMA~\cite{bhorkar2014performance}. However, this work does not address the fundamental issue of designing an LTE-based access scheme that is fair to legacy WiFi. In fact, results show that WiFi performance is substantially degraded in the presence of the more spectrum-efficient U-LTE scheme.

To the best of our knowledge, the work we present in this paper and our preliminary results in \cite{valls-commlet} are the first attempt to explore orthogonal channel access in WiFi bands. Based on this paradigm, we devise optimal transmission policies that maximize spectral efficiency of LBT-based LAA, while remaining harmless to native WiFi. Here, we substantially extend our initial work in \cite{valls-commlet}, as we specify optimal transmission policies in more general cases, including scenarios where transmitters use different modulations and/or loads. We further design an optimal transmission policy for synchronous LAA, which minimizes channel overhead when synchronizing to licensed LTE framing.

\section{Conclusions}\label{sec:conclusions}
To enable LTE deployment in unlicensed bands and seamless integration with existing cellular systems, 3GPP has recently specified an LBT-based solution named Licensed-Assisted Access (LAA). Despite its potential to attain \mbox{superior} user multiplexing and robustness (e.g., via Hybrid ARQ error recovery), in this paper we contended that the 3GPP LAA improves 802.11 MAC efficiency at the cost of penalizing incumbent WiFi networks, with 3GPP LAA configurations that are completely fair to WiFi achieiving inferior MAC performance compared to 802.11. We present a radically different approach to coexistence in unlicensed bands, which overcomes the limitations of 3GPP LAA and is compliant with the listen-before-talk requirement of, e.g., ETSI's EN 301 893 regulation. The proposed coexistence mechanism builds a symbiotic relationship between incumbent WiFi and U-LTE that creates orthogonal airtime blocks for each system, thethereby avoiding collisions between them and substantially increasing the MAC layer efficiency of both technologies. Based on this orthogonal access procedure, we derive	 optimal transmission policies, namely ORLA and OLAA, for asynchronous and synchronous systems respectively, which maximize U-LTE throughput, yet cause no harm to background WiFi networks. Finally, by means of extensive system-level simulations, we demonstrated the proposed transmission policies attain LBT throughput gains above 200\% with no negative impact on WiFi. 

\bibliographystyle{IEEEtran}


\end{document}